\documentclass[twocolumn]{IEEEtran}
\usepackage{epsfig}
\usepackage{amsmath}
\usepackage{amssymb}
\usepackage{verbatim}
\usepackage{epic}
\usepackage{graphicx}

\newcommand{\eg}{{\em e.g., }}

\newcommand{\sto}{\operatorname{s.t.} \,}
\newcommand{\ess}{\operatorname{ess}}
\newcommand{\sinc}{\operatorname{sinc}}

\newcommand{\ejo}{{{(e^{j\omega})}}}
\newcommand{\ejt}{{{(e^{j\omega T})}}}

\newcommand{\Span}{\operatorname{span} }

\newcommand{\I}{{\mathcal{I}}}

\newcommand{\A}{{\mathcal{A}}}

\newcommand{\bbf}{{{\bf F}}}
\newcommand{\bba}{{{\bf A}}}
\newcommand{\bbz}{{{\bf Z}}}
\newcommand{\bbm}{{{\bf M}}}
\newcommand{\bbw}{{{\bf W}}}

\newcommand{\bbu}{{{\bf U}}}
\newcommand{\bbv}{{{\bf V}}}

\newcommand{\bbd}{{{\bf D}}}
\newcommand{\bbc}{{{\bf C}}}
\newcommand{\bbx}{{{\bf X}}}

\newcommand{\bbi}{{{\bf I}}}

\newcommand{\bbq}{{{\bf Q}}}
\newcommand{\bc}{{{\bf c}}}
\newcommand{\bd}{{{\bf d}}}
\newcommand{\ba}{{{\bf a}}}

\newcommand{\bx}{{{\bf x}}}

\newcommand{\bs}{{{\bf s}}}

\newcommand{\bb}{{{\bf b}}}

\newcommand{\bgam}{\mbox{\boldmath{$\gamma$}}}
\newcommand{\bgaml}{\mbox{\boldmath{$\Gamma$}}}
\newcommand{\bpsi}{\mbox{\boldmath{$\Psi$}}}
\newcommand{\bphi}{\mbox{\boldmath{$\Phi$}}}

\newcommand{\bpsil}{\mbox{\boldmath{$\psi$}}}
\newcommand{\bphil}{\mbox{\boldmath{$\phi$}}}

\newcommand{\bl}{\left(}
\newcommand{\br}{\right)}
\newcommand{\blc}{\left\{}
\newcommand{\brc}{\right\}}

\newcommand{\ZZ}{{\mathbb{Z}}}
\newcommand{\CC}{{\mathbb{C}}}
\newcommand{\RR}{{\mathbb{R}}}

\newcommand{\inner}[2]{{\langle#1,#2\rangle}}

\newtheorem{theorem}{Theorem}

\newtheorem{proposition}{Proposition}

\title{Uncertainty Relations for Shift-Invariant Analog Signals}
\author{Yonina C. Eldar\thanks{Department of Electrical Engineering,
Technion---Israel Institute of Technology, Haifa 32000, Israel.
Phone: +972-4-8293256, fax: +972-4-8295757, E-mail:
yonina@ee.technion.ac.il. This work was supported in part by the
Israel Science Foundation under Grant no. 1081/07 and by the
European Commission in the framework of the FP7 Network of
Excellence in Wireless COMmunications NEWCOM++ (contract no.
216715).}}
\date{\today}

\begin{document}

\maketitle

\begin{abstract}

The past several years have witnessed a surge of research
investigating various aspects of sparse representations and
compressed sensing. Most of this work has focused on the
finite-dimensional setting in which the goal is to decompose a
finite-length vector into a given finite dictionary. Underlying
many of these results is the conceptual notion of an uncertainty
principle: a signal cannot be sparsely represented in two
different bases. Here, we extend these ideas and results to the
analog, infinite-dimensional setting by considering signals that
lie in a finitely-generated shift-invariant (SI) space. This class
of signals is rich enough to include many interesting special
cases such as multiband signals and splines. By adapting the
notion of coherence defined for finite dictionaries to infinite SI
representations, we develop an uncertainty principle similar in
spirit to its finite counterpart. We demonstrate tightness of our
bound by considering a bandlimited lowpass train that achieves the
uncertainty principle. Building upon these results and similar
work in the finite setting, we show how to find a sparse
decomposition in an overcomplete dictionary by solving a convex
optimization problem. The distinguishing feature of our approach
is the fact that even though the problem is defined over an
infinite domain with infinitely many variables and constraints,
under certain conditions on the dictionary spectrum our algorithm
can find the sparsest representation by solving a
finite-dimensional problem.
\end{abstract}

\section{Introduction}

Uncertainty relations date back to the work of Weyl and Heisenberg
who showed that a signal cannot be localized simultaneously in
both time and frequency. This basic principle was then extended by
Landau, Pollack, Slepian and later Donoho and Stark to the case in
which the signals are not restricted to be concentrated on a
single interval \cite{SP61,LP61,LP62,DS89}. The uncertainty
principle has deep philosophical interpretations. For example, in
the context of quantum mechanics it implies that a particle's
position and momentum cannot be simultaneously measured. In
harmonic analysis it imposes limits on the time-frequency
resolution \cite{G46}.

Recently, there has been a surge of research into discrete
uncertainty relations in more general finite-dimensional bases
\cite{DH01,EB02,CR06}. This work has been spurred in part by the
relationship between sparse representations and the emerging field
of compressed sensing \cite{D06,CRT06}. In particular, several
works have shown that discrete uncertainty relations can be used
to establish uniqueness of sparse decompositions in different
bases representations. Furthermore, there is an intimate
connection between uncertainty principles and the ability to
recover sparse expansions using convex programming
\cite{DH01,EB02,FN03}.

The vast interest in representations in redundant dictionaries
stems from the fact that the flexibility offered by such systems
can lead to decompositions that are extremely sparse, namely use
only a few dictionary elements. However, finding a sparse
expansion in practice is in general a difficult combinatorial
optimization problem. Two fundamental questions at the heart of
overcomplete representations are what is the smallest number of
dictionary elements needed to represent a given signal, and how
can one find the sparsest expansion in a computationally efficient
manner. In recent years, several key papers have addressed both of
these questions in a discrete setting, in which the signals to be
represented are finite-length vectors
\cite{DH01,EB02,FN03,GN03,DE03,T04,CRT06,CR06}.

The discrete generalized uncertainty principle for pairs of
orthonormal bases states that a vector in $\RR^N$ cannot be
simultaneously sparse in two orthonormal bases. The number of
non-zero representation coefficients is bounded below by the
inverse coherence \cite{DH01,EB02}. The coherence is defined as
the largest absolute inner product between vectors in each basis
\cite{MZ93,DH01}. This principle has been used to establish
conditions under which a convex $\ell_1$ optimization program can
recover the sparsest possible decomposition in a dictionary
consisting of both bases \cite{DH01,EB02,FN03}. These results
where later generalized in \cite{DE03,GN03,T04} to representations
in arbitrary dictionaries and to other efficient reconstruction
algorithms \cite{T04}.

The classical uncertainty principle is concerned with expanding a
continuous-time analog signal in the time and frequency domains.
However, the generalizations outlined above are mainly focused on
the finite-dimensional setting. In this paper, our goal is to
extend these recent ideas and results to the analog domain by
first deriving uncertainty relations for more general classes of
analog signals and arbitrary analog dictionaries, and then
suggesting concrete algorithms to decompose a continuous-time
signal into a sparse expansion in an infinite-dimensional
dictionary.

In our development, we focus our attention on continuous-time
signals that lie in shift-invariant (SI) subspaces of $L_2$
\cite{DDR94,GHM94,CE05}. Such signals can be expressed in terms of
linear combinations of shifts of a finite set of generators:
\begin{equation}
\label{eq:model} x(t)=\sum_{\ell=1}^N \sum_{n \in \ZZ}
a_\ell[n]\phi_\ell(t-nT),
\end{equation}
where $\phi_{\ell}(t), 1 \leq \ell \leq N$ are the SI generators,
and $a_{\ell}[n]$ are the expansions coefficients. Clearly, $x(t)$
is characterized by infinitely many coefficients $a_{\ell}[n]$.
Therefore, the finite results which provide bounds on the number
of non-zero expansion coefficients in pairs of bases
decompositions are not immediately relevant here. Instead, we
characterize analog sparsity as the number of active generators
that comprise a given representation, where the $\ell$th generator
is said to be active if $a_{\ell}[n],n \in \ZZ$ is not identically
zero.

Starting with expansions in two orthonormal bases, we show that
the number of active generators in each representation obeys an
uncertainty principle similar in spirit to that of finite
decompositions. The key to establishing this relation is in
defining an analog coherence between the two bases. Our definition
replaces the inner product in the finite setting by the largest
spectral value of the sampled cross-correlation between basis
elements, in the analog case. The similarity between the finite
and infinite cases can also be seen by examining settings in which
the uncertainty bound is tight. In the discrete scenario, the
lower uncertainty limit is achieved by decomposing a spike train
into the spike and Fourier bases, which are maximally incoherent
\cite{DS89}. To generalize this result to the analog domain we
first develop an analog spike-Fourier pair and prove that it is
maximally incoherent. The analog spike basis is obtained by
modulations of the basic lowpass filter (LPF), which is maximally
spread in frequency. In the time domain, these signals are given
by shifts of the sinc function, whose samples generate shifted
spikes. The discrete Fourier basis is replaced by an analog
Fourier basis, in which the elements are frequency shifts of a
narrow LPF in the continuous-time frequency domain. Tightness of
the uncertainty relation is demonstrated by expanding a train of
narrow LPFs in both bases.

We next address the problem of sparse decomposition in an
overcomplete dictionary, corresponding to using more than $N$
generators in (\ref{eq:model}). In the finite setting, it can be
shown that under certain conditions on the dictionary, a sparse
decomposition can be found using computationally efficient
algorithms such as $\ell_1$ optimization
\cite{CDS99,EB02,FN03,D06}. However, directly generalizing this
result to the analog setting is challenging. Although in principle
we can define an $\ell_1$ optimization program similar in spirit
to its finite counterpart, it will involve infinitely many
variables and constraints and therefore it is not clear how to
solve it in practice. Instead, we develop an alternative approach
by exploiting recent results on analog compressed sensing
\cite{ME07,ME08,E08,ME09}, that leads to a finite-dimensional
convex problem whose solution can be used to find the analog
sparse decomposition. Our algorithm is based on a three-stage
process: In the first step we sample the analog signal ignoring
the sparsity, and formulate the decomposition problem in terms of
sparse signal recovery from the given samples. In the second
stage, we exploit results on infinite measurement models (IMV) and
multiple measurement vectors (MMV) \cite{EM082,E08,Chen,Cotter} in
order to determine the active generators, by solving a
finite-dimensional convex optimization problem. Finally, we use
this information to simultaneously solve the resulting infinite
set of equations by inverting a finite matrix \cite{EM08}. Our
method works under certain technical conditions, which we
elaborate on in the appropriate section. We also indicate how
these results can be extended to more general classes of
dictionaries.

The paper is organized as follows. In Section~\ref{sec:prob} we
review the generalized discrete uncertainty principle and
introduce the class of analog signals we will focus on. The analog
uncertainty principle is formulated and proved in
Section~\ref{sec:ur}. In Section~\ref{sec:example} we consider a
detailed example illustrating the analog uncertainty relation and
its tightness. In particular we introduce the analog version of
the maximally incoherent spike-Fourier pair. Sparse decompositions
in two orthonormal analog bases are discussed in
Section~\ref{sec:rep}. These results are extended to arbitrary
dictionaries in Section~\ref{sec:frame}.

 In the sequel, we denote signals in
$L_2$ by lower case letters \eg $x(t)$, and SI subspaces of $L_2$
by $\A$. Vectors in $\CC^N$ are written as boldface lowercase
letters \eg $\bx$, and matrices as boldface uppercase letters \eg
$\bba$. The $i$th element of a vector $\bx$ is denoted $x_i$. The
identity matrix of appropriate dimension is written as $\bbi$. For
a given matrix $\bba$, $\bba^T$, $\bba^H$ are its transpose and
conjugate transpose respectively, $\bba_\ell$ is its $\ell$th
column, and $\bba^\ell$ is the $\ell$th row. The standard
Euclidean norm is denoted $\|\bx\|_2=\sqrt{\bx^H\bx}$,
$\|\bx\|_1=\sum_i |x_i|$ is the $\ell_1$ norm of $\bx$, and
$\|\bx\|_0$ is the cardinality of $\bx$ namely the number of
non-zero elements.
 The complex conjugate of a complex number $a$ is denoted
$\overline{a}$. The Fourier transform of a signal $x(t)$ in $L_2$
is defined as $X(\omega)=\int_{-\infty}^\infty x(t)e^{-j\omega t}
dt$. We use the convention that upper case letters represent
Fourier transforms. The discrete-time Fourier transform (DTFT) of
a sequence $x[n]$ in $\ell_2$ is defined by
$X(e^{j\omega})=\sum_{n=-\infty}^\infty x[n]e^{-j\omega n}$. To
emphasize the fact that the DTFT is $2\pi$-periodic we use the
notation $X(e^{j\omega})$.

\section{Problem Formulation}
\label{sec:prob}

\subsection{Discrete Uncertainty Principles}

The generalized uncertainty principle is concerned with pairs of
representations of a vector $\bx \in \RR^N$ in two different
orthonormal bases \cite{DH01,EB02}. Suppose we have two
orthonormal bases for $\RR^N$: $\{\bphil_{\ell},1 \leq \ell \leq
N\}$ and $\{\bpsil_{\ell},1 \leq \ell \leq N\}$. Any vector $\bx$
in $\RR^N$ can then be decomposed uniquely in terms of each one of
these vector sets:
\begin{equation}
\label{eq:xn} \bx=\sum_{\ell=1}^N a_\ell
\bphil_{\ell}=\sum_{\ell=1}^N b_\ell \bpsil_{\ell}.
\end{equation}
Since the bases are orthonormal, the expansion coefficients are
given by $a_{\ell}=\bphil_{\ell}^T\bx$ and
$b_{\ell}=\bpsil_{\ell}^T\bx$. Denoting by $\bphi,\bpsi$ the
matrices with columns $\bphil_{\ell},\bpsil_{\ell}$ respectively,
(\ref{eq:xn}) can be written as $\bx=\bphi\ba=\bpsi\bb$, with
$\ba=\bphi^T\bx$ and $\bb=\bpsi^T\bx$.

The uncertainty relation sets limits on the sparsity of the
decomposition for any vector $\bx \in \RR^N$. Specifically, let
$A=\|\ba\|_0$ and $B=\|\bb\|_0$ denote the number of non-zero
elements in each one of the expansions. The generalized
uncertainty principle \cite{EB02,DH01} states that
\begin{equation}
\label{eq:ucd} \frac{1}{2}\bl A+B \br \geq \sqrt{AB} \geq
\frac{1}{\mu(\bphi,\bpsi)},
\end{equation}
where $\mu(\bphi,\bpsi)$ is the coherence between the bases
$\bphi$ and $\bpsi$ and is defined by
\begin{equation}
\label{eq:mud}
\mu(\bphi,\bpsi)=\max_{\ell,r}|\bphil^T_{\ell}\bpsil_r|.
\end{equation}
The coherence measures the similarity between basis elements. This
definition was introduced in \cite{MZ93} to heuristically
characterize the performance of matching pursuit, and later used
in \cite{DH01,EB02,GN03,T04} in order to analyze the basis pursuit
algorithm.

It can easily be shown that $1/\sqrt{N} \leq \mu(\bphi,\bpsi) \leq
1$ \cite{DH01}. The upper bound follows from the Cauchy-Schwarz
inequality and the fact that the bases elements have norm $1$. The
lower bound is the result of the fact that the matrix
$\bbm=\bphi^T\bpsi$ is unitary and consequently
$\bbm^T\bbm=\bbi_N$. This in turn implies that the sum of the
squared elements of $\bbm$ is equal to $N$. Since there are $N^2$
variables, the value of the largest cannot be smaller than
$1/\sqrt{N}$. The lower bound of $1/\sqrt{N}$ can be achieved, for
example, by choosing the two orthonormal bases as the spike
(identity) and Fourier bases \cite{DS89}. With this choice, the
uncertainty relation (\ref{eq:ucd}) becomes
\begin{equation}
\label{eq:ucdf} A+B  \geq 2\sqrt{AB} \geq 2\sqrt{N}.
\end{equation}
Assuming $\sqrt{N}$ is an integer, the relations in
(\ref{eq:ucdf}) are all satisfied with equality when $\bx$ is a
spike train with spacing $\sqrt{N}$, resulting in $\sqrt{N}$
non-zero elements. This follows from the fact that the discrete
Fourier transform of $\bx$ is also a spike train with the same
spacing. Therefore, $\bx$ can be decomposed both in time and in
frequency into $\sqrt{N}$ basis vectors.

As we discuss in Section~\ref{sec:rep}, the uncertainty relation
provides insight into how sparse a signal $\bx$ can be represented
in an overcomplete dictionary consisting of $\bphi$ and $\bpsi$.
It also sheds light on the ability to compute such decompositions
using computationally efficient algorithms. Most of the research
to date on sparse expansions has focused on the discrete setting
in which the goal is to represent a finite-length vector $\bx$ in
$\RR^N$ in terms of a given dictionary using as few elements as
possible. First general steps towards extending the notions and
ideas underlying sparse representations and compressed sensing to
the analog domain have been developed in
\cite{ME07,E08,ME09,GE09}. Here we would like to take a further
step in this direction by extending the discrete uncertainty
principle to the analog setting.

\subsection{Shift-Invariant Signal Expansions}

In order to develop a general framework for analog uncertainty
principles we first need to describe the set of signals we
consider. A popular model in signal and image processing are
signals that lie in SI spaces. A finitely generated SI subspace in
$L_2$ is defined as \cite{DDR94,GHM94,CE05}:
\begin{equation}
\label{eq:si} \A=\blc x(t)=\sum_{\ell=1}^N \sum_{n \in \ZZ}
a_\ell[n]\phi_\ell(t-nT): a_\ell[n] \in \ell_2 \brc.
\end{equation}
The functions $\phi_\ell(t)$ are referred to as the generators of
$\A$. Examples of SI spaces include multiband signals
\cite{ME07,ME09} and spline functions \cite{S73b,EM08}. Expansions
of the type (\ref{eq:si}) are also encountered in communication
systems, when the analog signal is produced by pulse amplitude
modulation. In the Fourier domain, we can represent any $x(t) \in
\A$ as
\begin{equation}
\label{eq:xeq} X(\omega)=\sum_{\ell=1}^N A_\ell(e^{j\omega
T})\Phi_\ell(\omega),
\end{equation}
where
\begin{equation}
\label{eq:dft} A_\ell(e^{j\omega T})=\sum_{n \in \ZZ}
a_\ell[n]e^{-j\omega nT}
\end{equation}
is the DTFT of $a_{\ell}[n]$ at frequency $\omega T$, and is
$2\pi/T$ periodic.

In order to guarantee a unique stable representation of any signal
in $\A$ by a sequence of coefficients $a_\ell[n]$, the generators
$\phi_\ell(t)$ are typically chosen such that the functions
$\{\phi_{\ell}(t-nT),n \in \ZZ,1 \leq \ell \leq N\}$ form a Riesz
basis for $L_2$. This means that there exist constants $\alpha>0$
and $\beta<\infty$ such that
\begin{equation}
\label{eq:riesz} \alpha \|\ba\|^2 \leq \left\| \sum_{\ell=1}^N
\sum_{n \in \ZZ} a_\ell[n]\phi_\ell(t-nT)\right\|^2\leq \beta
\|\ba\|^2,
\end{equation}
where $\|\ba\|^2=\sum_{\ell=1}^N \sum_{n \in \ZZ} |a_\ell[n]|^2$,
and the norm in the middle term is the standard $L_2$ norm.
Condition (\ref{eq:riesz}) implies that any $x(t) \in \A$ has a
unique and stable representation in terms of the sequences
$a_\ell[n]$. By taking Fourier transforms in (\ref{eq:riesz}) it
follows that the shifts of the generators $\phi_\ell(t)$  form a
Riesz basis if and only if \cite{GHM94}
\begin{equation}
\label{eq:rc} \alpha \bbi \preceq \bbm_{\phi\phi}(e^{j\omega})
\preceq \beta \bbi,\quad \mbox{a.e. }  \omega,
\end{equation}
where
\begin{equation}
\label{eq:M} \bbm_{\phi\phi}(e^{j\omega})=\left[\begin{array}{ccc}
R_{\phi_1\phi_1}\ejo & \ldots &
R_{\phi_1\phi_N}\ejo\\
\vdots & \vdots &  \vdots \\
R_{\phi_N\phi_1}\ejo & \ldots & R_{\phi_N\phi_N}\ejo
\end{array}
\right],
\end{equation}
and for any two functions $\phi(t),\psi(t)$ with Fourier
transforms $\Phi(\omega),\Psi(\omega)$,
\begin{equation}
\label{eq:R} R_{\phi\psi}\ejo= \frac{1}{T}\sum_{k \in \ZZ}
\overline{\Phi}\bl\frac{\omega}{T}-\frac{2\pi}{T}k \br
\Psi\bl\frac{\omega}{T}-\frac{2\pi}{T}k \br.
\end{equation}
Note that $R_{\phi\psi}\ejo$ is the DTFT of the cross correlation
sequence $r_{\phi\psi}[n]=\inner{\phi(t-nT)}{\psi(t)}$, where the
inner product on $L_2$ is defined as
\begin{equation}
\inner{s(t)}{x(t)}=\int_{ -\infty}^\infty \overline{s}(t)x(t)dt.
\end{equation}

In Section~\ref{sec:frame} we consider overcomplete signal
expansions in which more than $N$ generators $\phi_{\ell}(t)$ are
used to represent a signal $x(t)$ in $\A$. In this case
(\ref{eq:riesz}) can be generalized to allow for stable
overcomplete decompositions in terms of a frame for $\A$. The
functions $\{\psi_\ell(t-nT),n \in \ZZ,1 \leq \ell \leq M\}$ form
a frame for the SI space $\A$ if there exist constants $\alpha>0$
and $\beta<\infty$ such that
\begin{equation}
\label{eq:frame} \alpha \|x(t)\|_2^2 \leq \sum_{\ell=1}^M \sum_{n
\in \ZZ} |\inner{\psi_\ell(t-nT)}{x(t)}|^2\leq \beta \|x(t)\|_2^2
\end{equation}
for all $x(t) \in \A$, where $\|x(t)\|_2^2=\inner{x(t)}{x(t)}$.

Our main interest is in expansions of a signal $x(t)$ in a SI
subspace $\A$ of $L_2$ in terms of orthonormal bases for $\A$. The
generators $\{\phi_{\ell}(t)\}$ of $\A$ form an orthonormal
basis\footnote{Here and in the sequel, when we say that a set of
signals $\{\phi_\ell(t)\}$ form (or generate) a basis, we mean
that the basis functions are $\{\phi_\ell(t-nT),n \in \ZZ,1 \leq
\ell \leq N\}$.} if
\begin{equation}
\label{eq:ortht}
\inner{\phi_\ell(t-nT)}{\phi_{r}(t-mT)}=\delta_{nm}\delta_{\ell
r},
\end{equation}
\sloppy for all $\ell,r,n,m$, where $\delta_{nm}=1$ if $n=m$ and
$0$
 otherwise.
Since \sloppy
$\inner{\phi_\ell(t-nT)}{\phi_{r}(t-mT)}=\inner{\phi_\ell(t-(n-m)T)}{\phi_{r}(t)}$,
(\ref{eq:ortht}) is equivalent to
\begin{equation}
\label{eq:orthtt}
\inner{\phi_\ell(t-nT)}{\phi_{r}(t)}=\delta_{n0}\delta_{\ell r}.
\end{equation}
Taking the Fourier transform of (\ref{eq:orthtt}), the
orthonormality condition can be expressed in the Fourier domain as
\begin{equation}
\label{eq:orth} R_{\phi_\ell\phi_{r}}\ejo=\delta_{\ell r}.
\end{equation}

Given an orthonormal basis $\{\phi_{\ell}(t-nT)\}$ for $\A$, the
unique representation coefficients $a_\ell[n]$ in (\ref{eq:si})
are given by $a_{\ell}[n]=\inner{\phi_{\ell}(t-nT)}{x(t)}$. This
can be seen by taking the inner product of $x(t)$ in (\ref{eq:si})
with $\phi_r(t-mT)$ and using the orthogonality relation
(\ref{eq:ortht}). Evidently, computing the expansion coefficients
in an orthonormal decomposition is straightforward. There is also
a simple relationship between the energy of $x(t)$ and the energy
of the coefficient sequence in this case, as incorporated in the
following proposition:
\begin{proposition}
\label{prop:orth} Let $\{\phi_{\ell}(t),1 \leq \ell \leq N\}$
generate an orthonormal basis for a SI subspace $\A$, and let
$x(t)=\sum_{\ell=1}^N \sum_{n \in \ZZ} a_\ell[n]\phi_\ell(t-nT)$.
Then
\begin{equation}
\label{eq:normo}
\|x(t)\|^2=\frac{T}{2\pi}\int_0^{\frac{2\pi}{T}}\sum_{\ell=1}^N
\left|A_{\ell}\ejt\right|^2 d\omega,
\end{equation}
where $\|x(t)\|_2^2=\inner{x(t)}{x(t)}$ and $A_{\ell}\ejo$ is the
DTFT of $a_{\ell}[n]$.
\end{proposition}
\begin{proof}
See Appendix~\ref{app:energy}.
\end{proof}

\subsection{Analog Problem Formulation}

In the finite-dimensional setting, sparsity is defined in terms of
the number of non-zero expansion coefficients in a given basis. In
an analog decomposition of the form (\ref{eq:model}), there are in
general infinitely many coefficients so that it is not immediately
clear how to define the notion of analog sparsity.

In our development, analog sparsity is measured by the number of
generators needed to represent $x(t)$. In other words, some of the
sequences $a_{\ell}[n]$ in (\ref{eq:model}) may be identically
zero, in which case
\begin{equation}
x(t)=\sum_{|\ell|=A} \sum_{n \in \ZZ} a_\ell[n]\phi_\ell(t-nT),
\end{equation}
where the notation $|\ell|=A$ means a sum over at most $A$
elements. Evidently, in our definition, sparsity is determined by
the energy of the entire sequence $a_{\ell}[n]$ and not by the
values of the individual elements.

In general, the number of zero sequences depends on the choice of
basis. Suppose we have an alternative representation
\begin{equation}
x(t)=\sum_{|\ell|=B} \sum_{n \in \ZZ} b_\ell[n]\psi_\ell(t-nT),
\end{equation}
where $\{\psi_{\ell}(t)\}$ also generate an orthonormal basis for
$\A$. An interesting question is whether there are limitations on
$A$ and $B$. In other words, can we have two representations that
are simultaneously sparse so that both $A$ and $B$ are small? This
question is addressed in the next section and leads to an analog
uncertainty principle, similar to (\ref{eq:ucd}). In
Section~\ref{sec:example} we prove that the relation we obtain is
tight, by constructing an example in which the lower limits are
satisfied.

As in the discrete setting we expect to be able to use fewer
generators in a SI expansion by allowing for an overcomplete
dictionary. In particular, if we expand $x(t)$ using both sets of
orthonormal bases we may be able to reduce the number of sequences
in the decomposition beyond what can be achieved using each basis
separately. The problem is how to find a sparse representation in
the joint dictionary in practice. Even in the discrete setting
this problem is NP-complete. However, results of
\cite{EB02,DE03,GN03,T04} show that under certain conditions a
sparse expansion can be determined by solving a convex
optimization problem. Here we have an additional essential
complication due to the fact that the problem is defined over an
infinite domain so that it has  infinitely many variables and
infinitely many constraints. In Section~\ref{sec:rep} we show that
despite the combinatorial complexity and infinite dimensions of
the problem, under certain conditions on the bases functions, we
can recover a sparse decomposition by solving a finite-dimensional
convex optimization problem.

\section{Uncertainty Relations in SI Spaces}
\label{sec:ur}
We begin by developing an analog of the discrete uncertainty
principle for signals $x(t)$ in SI subspaces. Specifically, we
show that the minimal number of sequences required to express
$x(t)$ in terms of any two orthonormal bases has to satisfy the
same inequality (\ref{eq:ucd}) as in the discrete setting, with an
appropriate modification of the coherence measure.
\begin{theorem}
\label{thm:uncertainty} Suppose we have a signal $x(t) \in \A$
where $\A$ is a SI subspace of $L_2$. Let $\{\phi_{\ell}(t),1 \leq
\ell \leq N\}$ and $\{\psi_{\ell}(t),1 \leq \ell \leq N\}$ denote
two orthonormal generators of $\A$, so that $x(t)$ can be
expressed in both bases with coefficient sequences
$a_{\ell}[n],b_{\ell}[n]$:
\begin{equation}
x(t)=\sum_{|\ell|=A} \sum_{n \in \ZZ} a_\ell[n]\phi_\ell(t-nT)=
\sum_{|\ell|=B} \sum_{n \in \ZZ} b_\ell[n]\psi_\ell(t-nT).
\end{equation}
Then,
\begin{equation}
\label{eq:uca} \frac{1}{2}(A+B)\geq \sqrt{AB} \geq
\frac{1}{\mu(\Phi,\Psi)},
\end{equation}
where
\begin{equation}
\label{eq:mu} \mu(\Phi,\Psi)=\max_{\ell,r}\ess \sup_{\omega}
\left| R_{\phi_{\ell}\psi_r}\ejo \right|,
\end{equation}
and $R_{\phi\psi}\ejo$ is defined by (\ref{eq:R}).
\end{theorem}

The coherence $\mu(\Phi,\Psi)$ of (\ref{eq:mu}) is a
generalization of the notion of discrete coherence (\ref{eq:mud})
defined for finite-dimensional bases. To see the analogy, note
that $R_{\phi\psi}\ejo$ is the DTFT of the correlation sequence
$r_{\phi\psi}[n]=\inner{\phi(t-nT)}{\psi(t)}$. On the other hand,
the finite-dimensional coherence can be written as
$\mu(\bphi,\bpsi)=(1/N)\max_{\ell,r}|\hat{\bphil}^H_{\ell}\hat{\bpsil}_r|$,
where $\hat{\bx}$ is the discrete Fourier transform of $\bx$ and
$N$ is the length of $\bx$.

\begin{proof}
Without loss of generality, we assume that $\|x(t)\|_2=1$. Since
$\{\phi_{\ell}(t)\}$ and $\{\psi_{\ell}(t)\}$ both generate
orthonormal bases, we have from Proposition~\ref{prop:orth} that
\begin{eqnarray}
\label{eq:normph} 1 & = & \frac{T}{2\pi}\int_0^{\frac{2\pi}{T}}
\sum_{|\ell|=A} |A_\ell\ejt|^2 d\omega \nonumber \\
& = & \frac{T}{2\pi}\int_0^{\frac{2\pi}{T}} \sum_{|\ell|=B}
|B_\ell\ejt|^2 d\omega.
\end{eqnarray}

Using the norm constraint and expressing $X(\omega)$ once in terms
of $\Phi_{\ell}(\omega)$ and once in terms of
$\Psi_{\ell}(\omega)$:
\begin{eqnarray}
\label{eq:mixed} \lefteqn{1 = \frac{1}{2\pi}\int_{-\infty}^\infty
|X(\omega)|^2d\omega} \nonumber \\
& = &
 \frac{1}{2\pi}\int_{-\infty}^\infty
\sum_{\substack{|\ell|=A \\ |r|=B}} \overline{A}_{\ell}\ejt
B_{r}\ejt \overline{\Phi}_{\ell}(\omega) \Psi_{r}(\omega) d\omega
\nonumber \\
&  =
&\frac{T}{2\pi}\int_0^{\frac{2\pi}{T}}\sum_{\substack{|\ell|=A \\
|r|=B}} \overline{A}_{\ell}\ejt B_{r}\ejt R_{\phi_{\ell}\psi_{r}}
\ejt
 d\omega \nonumber \\
& \leq & \frac{T}{2\pi}\int_0^{\frac{2\pi}{T}}
\sum_{\substack{|\ell|=A \\ |r|=B}}  \left| A_{\ell}\ejt\right|
\left| B_{r}\ejt \right|
\left|R_{\phi_{\ell}\psi_{r}}\ejt\right|  d\omega \nonumber \\
& \leq & \frac{\mu(\Phi,\Psi) T}{2\pi}\int_0^{\frac{2\pi}{T}}
\sum_{|\ell|=A} \left| A_{\ell}\ejt\right| \sum_{|r|=B} \left|
B_{r}\ejt \right|
 d\omega.
\end{eqnarray}
The third equality follows from rewriting the integral over the
entire real line as the sum of integrals over intervals of length
$2\pi/T$ as in (\ref{eq:integrals}) in Appendix~\ref{app:energy},
and the second inequality is a result of (\ref{eq:mu}).
 Applying the Cauchy-Schwarz inequality
to the integral in (\ref{eq:mixed}) we have
\begin{eqnarray}
\label{eq:csm} \lefteqn{\left[\int_0^{\frac{2\pi}{T}}
\sum_{|\ell|=A} \left| A_{\ell}\ejt\right| \sum_{|r|=B} \left|
B_{r}\ejt \right|
 d\omega\right]^{2}} \nonumber \\  && \hspace*{-0.2in}\leq  \int_0^{\frac{2\pi}{T}} \bl \sum_{|\ell|=A}
\left| A_{\ell}\ejt\right|\br^2
 d\omega
\int_0^{\frac{2\pi}{T}} \bl \sum_{\ell=1}^B \left|
B_{\ell}\ejt\right|\br^2
 d\omega .
\end{eqnarray}
Using the same inequality we can upper bound the sum in
(\ref{eq:csm}):
\begin{equation}
\bl\sum_{|\ell|=A} \left| A_{\ell}\ejt\right|\br^2
 \leq A \sum_{|\ell|=A}
\left| A_{\ell}\ejt\right|^2.
\end{equation}
Combining with (\ref{eq:csm}), (\ref{eq:mixed}) and
(\ref{eq:normph}) leads to
\begin{eqnarray}
1 \leq  \mu(\Phi,\Psi) \sqrt{AB}.
\end{eqnarray}
Using the well-known relation between the arithmetic and geometric
means completes the proof.
\end{proof}

An interesting question is how small $\mu(\Phi,\Psi)$ can be made
by appropriately choosing the bases. From
Theorem~\ref{thm:uncertainty} the smaller $\mu(\Phi,\Psi)$, the
stronger the restriction on the sparsity in both decompositions.
As we will see in Section~\ref{sec:rep}, such a limitation is
helpful in recovering the true sparse coefficients. In the finite
setting we have seen that $1/\sqrt{N} \leq \mu(\bphi,\bpsi) \leq
1$ \cite{DH01}. The next theorem shows that the same bounds hold
in the analog case.
\begin{theorem}
\label{thm:mu} Let $\{\phi_{\ell}(t),1 \leq \ell \leq N\}$ and
$\{\psi_{\ell}(t),1 \leq \ell \leq N\}$ denote two orthonormal
generators of a SI subspace $\A \subset L_2$ and let
$\mu(\Phi,\Psi)=\max_{\ell,r}\ess \sup \left|
R_{\phi_{\ell}\psi_r}\ejo \right|$, where $R_{\phi\psi}\ejo$ is
defined by (\ref{eq:R}). Then
\begin{equation}
\label{eq:bound}  \frac{1}{\sqrt{N}} \leq \mu(\Phi,\Psi) \leq 1.
\end{equation}
\end{theorem}
\begin{proof}
We begin by proving the upper bound, which follows immediately
from the Cauchy-Schwarz inequality and the orthonormality of the
bases:
\begin{equation}
\left|R_{\phi_{\ell}\psi_r}\ejo \right|\leq \bl
R_{\phi_{\ell}\phi_{\ell}}\ejo R_{\psi_r\psi_r}\ejo\br^{1/2}=1,
\end{equation}
where the last equality is a result of (\ref{eq:orth}). Therefore,
$\mu(\Phi,\Psi) \leq 1$.

To prove the lower bound, note that since $\phi_{\ell}(t)$ is in
$\A$ for each $\ell$, we can express it as
\begin{equation}
\phi_{\ell}(t)=\sum_{r=1}^N \sum_{n \in \ZZ}
a_r^{\ell}[n]\psi_r(t-nT)
\end{equation}
for some coefficients $a_r^{\ell}[n]$, or in the Fourier domain,
\begin{equation}
\label{eq:mup} \Phi_{\ell}(\omega)=\sum_{r=1}^N  A_r^{\ell}\ejt
\Psi_r(\omega).
\end{equation}
Since $\|\phi_\ell(t)\|=1$ and $\{\psi_r(t)\}$ are orthonormal, we
have from Proposition~\ref{prop:orth} that
\begin{equation}
\label{eq:an} \frac{T}{2\pi}\int_0^{\frac{2\pi}{T}} \sum_{r=1}^N
\left| A_{r}^{\ell}\ejt\right|^2
 d\omega=1,\quad 1 \leq \ell \leq N.
\end{equation}
Now, using (\ref{eq:mup}) and the orthonormality condition
(\ref{eq:orth}) it follows that
\begin{eqnarray}
\label{eq:mub} \mu(\Phi,\Psi) & \geq &
\left|R_{\phi_{\ell}\psi_{r}}\ejo
\right| \nonumber \\
&= &\left|  \sum_{s=1}^N \overline{A}_s^{\ell}\ejo
R_{\psi_{s}\psi_{r}}\ejo \right|= \left|A_r^{\ell}\ejo\right|.
\end{eqnarray}
Therefore,
\begin{eqnarray}
\label{eq:snorm}
\lefteqn{\int_{0}^{2\pi}\sum_{\ell,r=1}^N\left|R_{\phi_{\ell}\psi_{r}}\ejo
\right|^2 d\omega =} \nonumber \\
& = &
\sum_{\ell=1}^N\int_{0}^{2\pi}\sum_{r=1}^N\left|A_{r}^\ell\ejo
\right|^2 d\omega =2\pi N,
\end{eqnarray}
where the last equality follows from (\ref{eq:an}) by performing a
change of variables $\omega'=\omega T$ in the integral.
 If
$\mu(\Phi,\Psi)<1/\sqrt{N}$, then,
$|R_{\phi_{\ell}\psi_{r}}\ejo|<1/\sqrt{N}$ a.e. on $\omega$ and
\begin{equation}
\int_{0}^{2\pi}\sum_{\ell,r=1}^N\left|R_{\phi_{\ell}\psi_{r}}\ejo
\right|^2  d\omega <2\pi N,
\end{equation}
which contradicts (\ref{eq:snorm}).
\end{proof}

It is easy to see that the lower bound in (\ref{eq:bound}) is
achieved  if $R_{\phi_{\ell}\psi_{r}}\ejo=1/\sqrt{N}$ for all
$\ell,r$ and $\omega$. In this case the uncertainty relation
(\ref{eq:uca}) becomes
\begin{equation}
\label{eq:ucl} A+B \geq 2\sqrt{AB} \geq 2\sqrt{N}.
\end{equation}
As discussed in Section~\ref{sec:prob}, in the discrete setting
with $\sqrt{N}$ an integer, the inequalities in (\ref{eq:ucl}) are
achieved using the spike-Fourier basis and $\bx$ equal to a spike
train. In the next section we show that equality in (\ref{eq:ucl})
can be satisfied in the analog case as well using a pair of bases
that is analogous to the spike-Fourier pair, and a bandlimited
signal $x(t)$ equal to a lowpass train.

\section{Achieving the Uncertainty Principle} \label{sec:example}

\subsection{Minimal Coherence} \label{sec:coherence}

Consider the space $\A$ of real signals bandlimited to $(-\pi
N/T,\pi N/T]$. As we show below, any signal in $\A$ can be
expressed in terms of $N$ SI generators. We would like to choose
two orthonormal bases, analogous to the spike-Fourier pair in the
finite setting, for which the coherence achieves its lower limit
of $1/\sqrt{N}$. To this end, we first highlight the essential
properties of the finite spike-Fourier bases in $\CC^N$, and then
choose an analog pair with similar characteristics.

The basic properties of the spike-Fourier pair are illustrated in
Fig.~\ref{fig:ifd}. The first element of the spike basis,
$\bphil_1$, is equal to a constant in the discrete Fourier domain,
as illustrated in the left-hand side of Fig.~\ref{fig:ifd}. The
remaining basis vectors are generated by shifts in time, or
modulations in frequency, as depicted in the bottom part of the
figure. In contrast, the first vector of the Fourier basis is
sparse in frequency: it is represented by a single frequency
component as illustrated in the right-hand side of the figure. The
rest of the basis elements are obtained by shifts in frequency.
\begin{figure*}
\centering
\includegraphics[scale=.8]{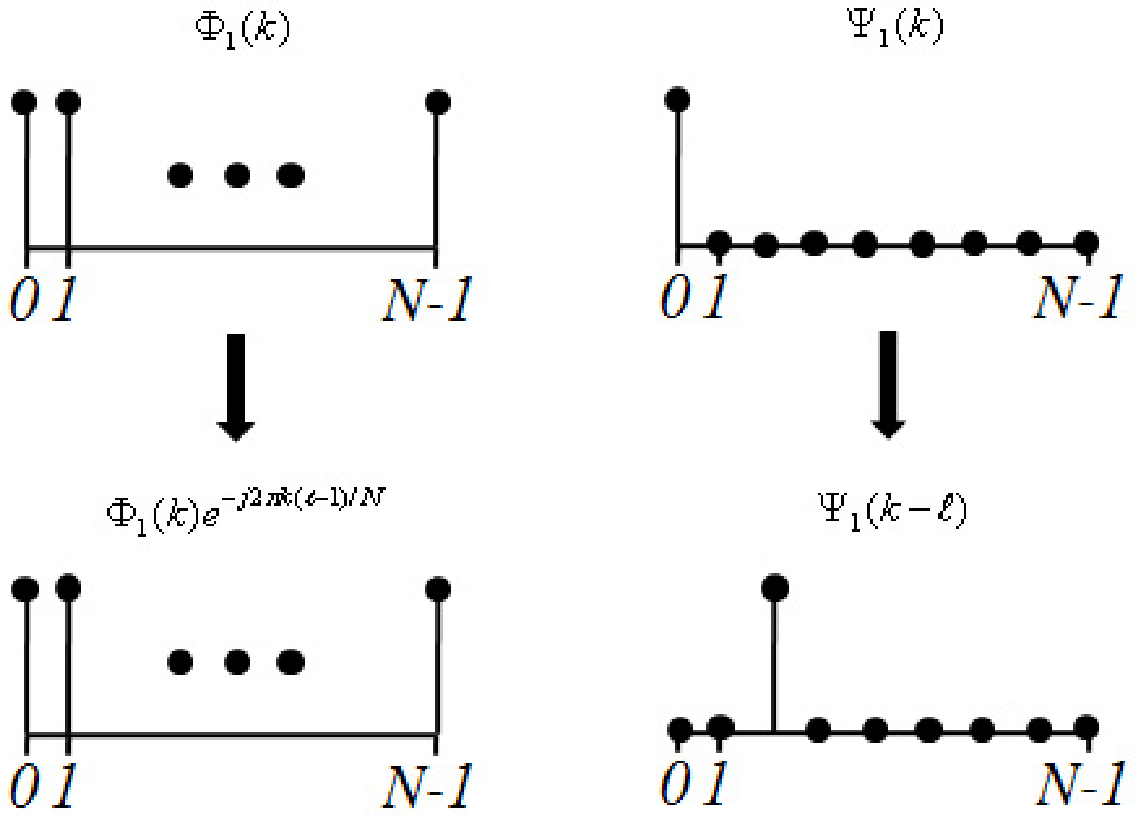}
\caption{Discrete Fourier-domain representation of the
spike-Fourier bases in $\CC^N$. The left-hand side is the discrete
Fourier transform of the spike basis. The right-hand side
represents the discrete Fourier transform of the Fourier basis.
The top row corresponds to the first basis function, while the
bottom row represents the $\ell$th basis function.}
\label{fig:ifd}
\end{figure*}

We now construct two orthonormal bases for $\A$ with minimal
coherence by mimicking these properties in the continuous-time
Fourier domain. Since we are considering the class of signals
bandlimited to $\pi N/T$, we only treat this frequency range. As
we have seen, the basic element of the spike basis occupies the
entire frequency spectrum. Therefore, we choose our first analog
generator $\phi_{1}(t)$ to be constant over the frequency range
$(-\pi N/T,\pi N/T]$.  The remaining generators are obtained by
shifts in time of $\phi_{1}(t)$ or modulations in frequency:
\begin{equation}
\label{eq:psie} \Phi_{\ell}(\omega)=\left\{
\begin{array}{ll}
\sqrt\frac{T}{N} e^{-j \omega (\ell-1) T/N}, & \omega \in (-\pi N/T,\pi N/T]; \\
0, & \mbox{otherwise},
\end{array}
\right.
\end{equation}
corresponding to
\begin{equation}
\label{eq:psit} \phi_{\ell}(t)=\sqrt{\frac{N}{T}}\sinc((t-(\ell-1)
T')/T'),
\end{equation}
with $T'=T/N$. The normalization constant is chosen to ensure that
the basis vectors have unit norm. With slight abuse of
terminology, we refer to the set $\{\phi_{\ell}(t),1 \leq \ell
\leq N\}$ as the analog spike basis (the basis is actually
constructed by shifts of this set with period $T$). Note that the
samples of $\phi_{\ell}(t)$ at times $nT'$ create a shifted spike
sequence, further justifying the analogy. The Fourier transform of
the analog spike basis is illustrated in the left-hand side of
Fig.~\ref{fig:ifa}.
\begin{figure*}
\centering
\includegraphics[scale=.8]{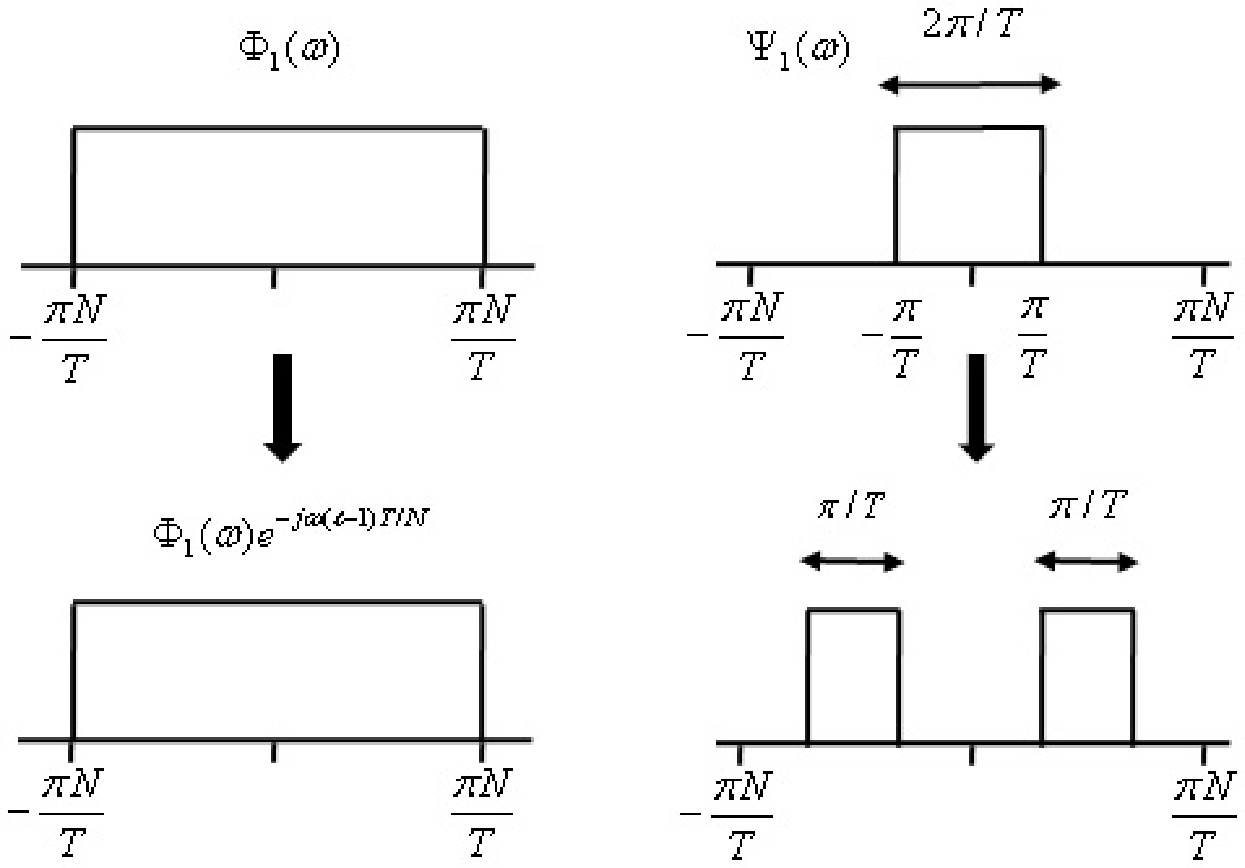}
\caption{Continuous Fourier-domain representation of the analog
spike-Fourier bases in $\A$. The left-hand side is the Fourier
transform of the spike basis. The right-hand side represents the
Fourier transform of the Fourier basis. The top row corresponds to
the first generator, while the bottom row represents the $\ell$th
generator.}\label{fig:ifa}
\end{figure*}

To construct the second orthonormal basis, we choose $\psi_1(t)$
to be sparse in frequency, as in the finite case. The remaining
generators are obtained by shifts in frequency. To ensure that the
generators are real we must have that
$\Psi_{\ell}(\omega)=\overline{\Psi}_{\ell}(-\omega)$. Therefore,
we consider only the interval $[0,\pi N/T]$. Since we have $N$
real generators, we divide this interval into equal sections of
length $\pi/T$, and choose each $\Psi_{\ell}(\omega)$ to be
constant over the corresponding interval, as illustrated in
Fig.~\ref{fig:ifa}. More specifically, let
\begin{equation}
\I_{\ell}=\{\omega:|\omega| \in (\pi (\ell-1)/T,\pi \ell/T]\},
\end{equation}
be the $\ell$th interval. Then
\begin{equation}
\label{eq:phie} \Psi_{\ell}(\omega)=\left\{
\begin{array}{ll}
\sqrt{T}, & \omega \in \I_{\ell}; \\
0, & \mbox{otherwise}.
\end{array}
\right.
\end{equation}

The analog pair of bases generated by
$\{\Phi_{\ell}(\omega),\Psi_{\ell}(\omega),1 \leq \ell \leq N\}$
is referred to as the analog spike-Fourier pair. In order to
complete the analogy with the discrete spike-Fourier bases we need
to show that both analog sets are orthonormal and generate $\A$,
and that their coherence is equal to $1/\sqrt{N}$. The latter
follows immediately by noting that
\begin{equation}
\overline{\Phi}_{\ell}(\omega)\Psi_{r}(\omega)=\left\{
\begin{array}{ll}
\frac{T}{\sqrt{N}} e^{j \omega (\ell-1) T/N}, & \omega \in \I_{r}; \\
0, & \mbox{otherwise}.
\end{array}
\right.
\end{equation}
It is easy to see that replicas of $\I_r$ at distance $2\pi/T$
will not overlap. Furthermore, these replicas tile the entire
frequency axis; therefore,
$|R_{\phi_{\ell}\psi_{r}}\ejo|=1/\sqrt{N}$, and
$\mu(\Phi,\Psi)=1/\sqrt{N}$.

To show that $\{\psi_{\ell}(t),1 \leq \ell \leq N\}$ generate
$\A$, note that any $x(t) \in \A$ can be expressed in the form
(\ref{eq:si}) (or (\ref{eq:xeq}))  by choosing $A_{\ell}\ejt
=X(\omega)$ for $\omega \in \I_{\ell}$. If $X(\omega)$ is zero on
one of the intervals $\I_{\ell}$, then $A_{\ell}\ejo$ will also be
zero, leading to the multiband structure studied in
\cite{ME07,ME09}. Since the intervals on which
$\Psi_{\ell}(\omega)$ are non-zero do not overlap, the basis is
also orthogonal. Finally, orthonormality follows from our choice
of scaling.

Proving that $\{\phi_{\ell}(t),1 \leq \ell \leq N\}$ generate an
orthonormal basis is a bit more tricky. To see that these
functions span $\A$ note that from Shannon's sampling theorem, any
function $x(t)$ bandlimited to $\pi /T'$ with $T'=T/N$ can be
written as
\begin{equation}
\label{eq:xshannon} x(t)=\sum_{n \in \ZZ} x(nT')\sinc((t-nT')/T').
\end{equation}
Substituting $n=mN+\ell-1$, we can replace the sum over $n$ by the
double sum over $m \in \ZZ$ and $1 \leq \ell \leq N$, resulting in
\begin{eqnarray}
\label{eq:xshannon2} x(t) & = & \sum_{\ell=1}^N \sum_{m \in \ZZ}
a_{\ell}[m]\sinc((t-(\ell-1) T' -mT))/T') \nonumber
\\
& = & \sqrt{\frac{T}{N}}\sum_{\ell=1}^N \sum_{n \in \ZZ}
a_{\ell}[n]\phi_{\ell}(t-nT),
\end{eqnarray}
with $a_{\ell}[n]=x((\ell-1)T' +nT)$, proving that
$\{\phi_{\ell}(t)\}$ generate $\A$. Orthonormality of the basis
follows from
\begin{equation}
R_{\phi_{\ell}\phi_{r}}\ejo=\frac{1}{N}e^{j\omega(\ell-r)/N}\sum_{k=0}^{N-1}
 e^{-j2\pi k(\ell-r)/N}=\delta_{r\ell},
\end{equation}
where we used the relation
\begin{equation}
\label{eq:drl} \sum_{k=0}^{N-1}
 e^{-j2\pi k (\ell-r)/N}=N\delta_{r\ell}.
\end{equation}

\subsection{Tightness of the Uncertainty Relation}

Given any signal $x(t)$ in $\A$, the uncertainty relation for the
analog spike-Fourier pair states that the number of non-zero
sequences in the spike and Fourier bases must satisfy
(\ref{eq:ucl}). We now show that when $\sqrt{N}$ is an integer,
these inequalities can be achieved with equality with an
appropriate choice of $x(t)$, so that the uncertainty principle is
tight. To determine such a signal $x(t)$, we again mimic the
construction in the discrete case.

As we discussed in Section~\ref{sec:prob}, when using the finite
Fourier-spike pair, we have equalities in (\ref{eq:ucl}) when $\bx
\in \RR^N$ is a spike train with $\sqrt{N}$ non-zero values,
equally spaced, as illustrated in the left-hand side of
Fig.~\ref{fig:uce}. This follows from the fact that the spike
train has the same form in both time and frequency. To construct a
signal in $\A$ satisfying the analog uncertainty relation, we
replace each Fourier-domain spike in the discrete setting by a
shifted LPF of width $2\pi/T$ in the analog Fourier domain. To
ensure that there are $\sqrt{N}$ non-zero intervals of length
$2\pi/T$ in $(-\pi N/T,\pi N/T]$, the frequency spacing between
the LPFs is set to $2 \pi \sqrt{N}/T$, as depicted in the
right-hand side of Fig.~\ref{fig:uce}. This signal can be
represented in frequency by $\sqrt{N}$ basis functions
$\Psi_{m}(\omega)$, with $m=2\sqrt{N}\ell,1 \leq \ell \leq \lfloor
\sqrt{N}/2 \rfloor$, and $m=2\sqrt{N}(\ell-1)+1,1 \leq \ell \leq
\lceil \sqrt{N}/2 \rceil$.
 It
therefore remains to be shown that $x(t)$ can also be expanded in
time using $\sqrt{N}$ signals $\phi_{m}(t)$.
\begin{figure*}
\centering
\includegraphics[scale=.8]{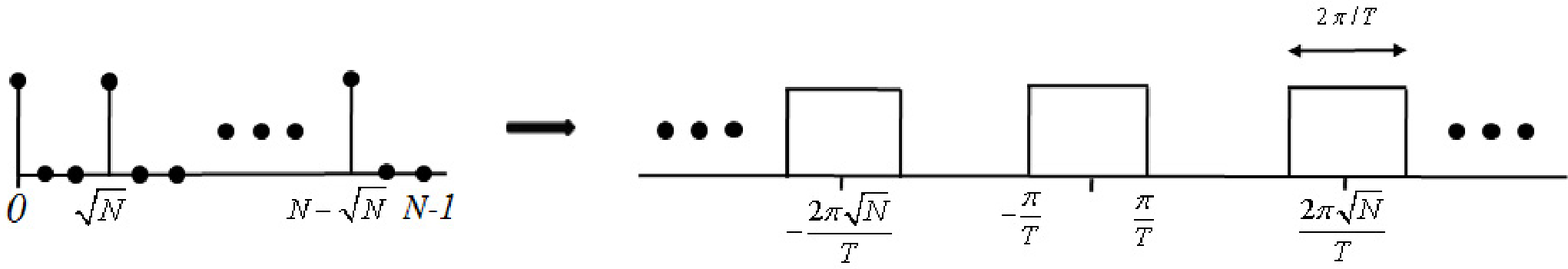}
\caption{Discrete and analog signals satisfying the uncertainty
principle with equality. The left-hand side is the discrete
Fourier transform of the spike train. The right-hand side
represents the analog Fourier transform of the LPF
train.}\label{fig:uce}
\end{figure*}

Since $x(t)$ is bandlimited to $\pi N/T$,
\begin{equation}
x(t)= \sum_{\ell=1}^N \sum_{n \in \ZZ} a_\ell[n]\phi_\ell(t-nT),
\end{equation}
where $a_\ell[n]=\inner{\phi_\ell(t-nT)}{x(t)}$. In the Fourier
domain we have
\begin{equation}
A_\ell\ejo=
 \frac{e^{j\omega(\ell-1)/N}}{\sqrt{NT}}\sum_{k \in \ZZ}
e^{-j\omega 2\pi k(\ell-1)/N} X\bl\frac{\omega}{T}-\frac{2\pi}{T}k
\br.
\end{equation}
Due to the fact that $a_\ell[n]$ is a real sequence,
$A_\ell\ejo=\overline{A}_\ell(e^{-j\omega})$. Therefore we
consider $A_\ell\ejo$ on the interval $[0, \pi]$. For values of
$\omega$ in this interval, $X(\omega/T-2\pi k/T)$ is non-zero only
for indices $ k=m\sqrt{N}$ with $\lfloor-\sqrt{N}/2+1\rfloor\leq m
\leq \lfloor \sqrt{N}/2\rfloor$. Thus,
\begin{eqnarray}
A_\ell\ejo & = &
 \frac{e^{j\omega(\ell-1)/N}}{\sqrt{NT}}\sum_{m=\lfloor-\sqrt{N}/2+1\rfloor}^{\lfloor \sqrt{N}/2\rfloor}
e^{-j\omega 2\pi m(\ell-1)/\sqrt{N}} \nonumber \\
& = &
\frac{e^{j\omega(\ell-1)/N}}{\sqrt{T}}\delta_{\ell-1,r\sqrt{N}},
\end{eqnarray}
where $r$ is an arbitrary integer. The last equality follows from
(\ref{eq:drl}) and the fact that the sum is over $\sqrt{N}$
consecutive values. Since $1 \leq \ell \leq N$, $A_\ell\ejo$ is
nonzero for $\sqrt{N}$ indices $\ell$, so that $x(t)$ can be
expanded in terms of $\sqrt{N}$ generators $\phi_{\ell}(t)$.

\section{Recovery of Sparse Representations}
\label{sec:rep}

\subsection{Discrete Representations}

One of the important implications of the discrete uncertainty
principle is its relation to sparse approximations
\cite{DH01,EB02,DE03,T04}. Given two orthonormal bases
$\bphi,\bpsi$ for $\RR^N$ an interesting question is whether one
can reduce the number of non-zero expansion coefficients required
to represent a vector $\bx \in \RR^N$ by decomposing it in terms
of the concatenated dictionary
\begin{equation}
\label{eq:df} \bbd=\left[ \begin{array}{cc} \bphi & \bpsi
\end{array}
\right].
\end{equation}
In many cases such a representation can be much sparser than the
decomposition in either of the bases alone. The difficulty is in
actually finding a sparse expansion $\bx=\bbd\bgam$ in which
$\bgam$ has as few non-zero components as possible. Since $\bbd$
has more columns than rows, the set of equations $\bx=\bbd\bgam$
is underdetermined and therefore $\bx$ can have multiple
representations $\bgam$. Finding the sparsest choice can be
translated into the combinatorial optimization problem
\begin{equation}
\label{eq:lo}
 \min_{\bgam}\|\bgam\|_0 \quad \sto  \bx=\bbd\bgam.
\end{equation}
Problem (\ref{eq:lo}) is NP-complete in general and cannot be
solved efficiently. A surprising result of \cite{DH01,EB02,FN03},
summarized below in Proposition~\ref{prop:unique}, is that if the
coherence $\mu(\bphi,\bpsi)$ between the two bases is small enough
with respect to the sparsity of $\bgam$, then the sparsest
possible $\bgam$ is unique and can be found by the basis pursuit
algorithm. This algorithm is a result of replacing the non-convex
$\ell_0$ norm by the convex $\ell_1$ norm:
\begin{equation}
\label{eq:l1d} \min_{\bgam}\|\bgam\|_1 \quad \sto  \bx=\bbd\bgam.
\end{equation}
\begin{proposition}
\label{prop:unique}  Let $\bbd=[\bphi\,\,\bpsi]$ be a dictionary
consisting of two orthonormal bases with coherence
$\mu(\bphi,\bpsi)=\max_{\ell,r}|\bphil^T_{\ell}\bpsil_r|$. If a
vector $\bx$ has a sparse decomposition in $\bbd$ such that
$\bx=\bbd\bgam$ and $\|\bgam\|_0<1/\mu(\bphi,\bpsi)$ then this
representation is unique, namely there cannot be another $\bgam'$
with $\|\bgam'\|_0<1/\mu(\bphi,\bpsi)$ and $\bx=\bbd\bgam'$.
Furthermore, if
\begin{equation}
\label{eq:gamc} \|\bgam\|_0 <
\frac{\sqrt{2}-0.5}{\mu(\bphi,\bpsi)},
\end{equation}
then the unique sparse representation can be found by solving the
$\ell_1$ optimization problem (\ref{eq:l1d}).
\end{proposition}
As detailed in \cite{DH01,EB02}, the proof of
Proposition~\ref{prop:unique} follows from the generalized
discrete uncertainty principle.

Another useful result on dictionaries with low coherence is that
every set of $k \leq 2/\mu(\bphi,\bpsi)-1$ columns are linearly
independent \cite[Theorem 6]{DE03}. This result can be stated in
terms of the Kruskal-rank of $\bbd$ \cite{Kruskal}, which is the
maximal number $q$ such that every set of $q$ columns of $\bbd$ is
linearly independent.
\begin{proposition}\cite[Theorem 6]{DE03}
\label{prop:kr} Let $\bbd=[\bphi\,\,\bpsi]$ be a dictionary
consisting of two orthonormal bases with coherence
$\mu(\bphi,\bpsi)$. Then $\sigma(\bbd) \geq 2/\mu(\bphi,\bpsi)-1$
where $\sigma(\bbd)$ is the Kruskal rank of $\bbd$.
\end{proposition}


\subsection{Analog Representations}
\label{sec:ar}

We would now like to generalize these recovery results to the
analog setup. However, it is not immediately clear how to extend
the finite $\ell_1$ basis pursuit algorithm of (\ref{eq:l1d}) to
the analog domain.

To set up the analog sparse decomposition problem, suppose we have
a signal $x(t)$ that lies in a space $\A$, and let
$\{\phi_{\ell}(t),1 \leq \ell \leq N\},\{\psi_{\ell}(t),1 \leq
\ell \leq N\}$ be two orthonormal generators of $\A$. Our goal is
to represent $x(t)$ in terms of the joint dictionary
$\{d_{\ell}(t-nT),1 \leq \ell \leq 2N\}$ with
\begin{equation}
\label{eq:dl} d_{\ell}(t)=\left\{
\begin{array}{ll}
\phi_{\ell}(t), & 1 \leq \ell \leq N; \\
\psi_{\ell-N}(t), &  N+1 \leq \ell \leq 2N,
\end{array}
\right.
\end{equation}
using as few non-zero sequences as possible. Denoting by
$\bgam[n]$ the vector at point-$n$ whose elements are
$\gamma_{\ell}[n]$, our problem is to choose the vector sequence
$\bgam[n]$ such that
\begin{equation}
\label{eq:decoma} x(t)=\sum_{\ell=1}^{2N} \sum_{n \in \ZZ}
\gamma_\ell[n]d_\ell(t-nT),
\end{equation}
and $\gamma_\ell[n]$ is identically zero for the largest possible
number of indices $\ell$.

We can count the number of non-zero sequences by first computing
the $\ell_2$-norm of each sequence. Clearly, $\gamma_\ell[n]$ is
equal $0$ for all $n$ if and only if its $\ell_2$ norm
$\|\gamma_{\ell}[n]\|_2=(\sum_n |\gamma_{\ell}^2[n]|)^{1/2}$ is
zero. Therefore, the number of non-zero sequences
$\gamma_{\ell}[n]$ is equal to $\|\bc\|_0$ where
$c_\ell=\|\gamma_{\ell}[n]\|_2$. For ease of notation, we denote
$\|\bgam\|_{2,0}=\|\bc\|_0$, and similarly
$\|\bgam\|_{2,1}=\|\bc\|_1$. Finding the sparsest decomposition
(\ref{eq:decoma}) can then be written as
\begin{equation}
\label{eq:loa}
 \min_{\bgam}\|\bgam\|_{2,0} \quad \sto  x(t)=\sum_{\ell=1}^{2N} \sum_{n \in \ZZ}
\gamma_\ell[n]d_\ell(t-nT).
\end{equation}
Problem (\ref{eq:loa}) is the analog version of (\ref{eq:lo}).
However, in addition to being combinatorial as its finite
counterpart, (\ref{eq:loa}) also has infinitely many variables and
constraints.

In order to extend the finite-dimensional decomposition results to
the analog domain, there are two main questions we need to
address:
\begin{enumerate}
\item  Is there a unique sparse representation for any input
signal in a given dictionary? \item How can we compute a sparse
expansion in practice, namely solve (\ref{eq:loa}), despite the
combinatorial complexity and infinite dimensions?
\end{enumerate}
The first problem is easy to answer. Indeed, the uniqueness
condition of Proposition~\ref{prop:unique} can be readily extended
to the analog case. This is due to the fact that its proof is
based on the uncertainty relation (\ref{eq:ucd}) which is
identical to (\ref{eq:uca}), with the appropriate modification to
the coherence measure.
\begin{proposition}
\label{prop:uniquea} Suppose that  a signal $x(t) \in \A$ has a
sparse representation in the joint dictionary $\{d_{\ell}(t-nT),n
\in \ZZ,1 \leq \ell \leq 2N\}$ of (\ref{eq:dl}) which consists of
two orthonormal bases $\{\phi_{\ell}(t-nT),\psi_{\ell}(t-nT),n \in
\ZZ,1 \leq \ell \leq N\}$. If the coefficient sequences
$\gamma_\ell[n]$ of (\ref{eq:decoma}) satisfy
$\|\bgam\|_{2,0}<1/\mu(\Phi,\Psi)$ where $\mu(\Phi,\Psi)$ is the
coherence defined by (\ref{eq:mu}), then this representation is
unique.
\end{proposition}

The second, more difficult question, is how to find a unique
sparse representation when it exists. We may attempt to develop a
solution by replacing the $\ell_0$ norm in (\ref{eq:loa})  by  an
$\ell_1$ norm, as in the finite-dimensional case. This leads to
the convex program
\begin{equation}
\label{eq:l1a} \min_{\bgam}\|\bgam\|_{2,1} \quad \sto
x(t)=\sum_{\ell=1}^{2N} \sum_{n \in \ZZ}
\gamma_\ell[n]d_\ell(t-nT).
\end{equation}
However, in practice, it is not clear how to solve (\ref{eq:l1a})
since it is defined over an infinite set of variables
$\gamma_{\ell}[n]$, and has infinitely many constraints (for all
$t$).

Our approach to treating the analog decomposition problem is to
first sample the signal $x(t)$ at a high enough rate, so that
$x(t)$ can be determined from the given samples. We will then show
that the decomposition problem can be recast in the Fourier domain
as that of recovering a set of sparse vectors that share a joint
sparsity pattern, from the given sequences of samples. The
importance of this reformulation is that under appropriate
conditions, it allows to determine the joint support set (or the
active generators) by solving a finite-dimensional optimization
problem. Once the active generators are determined, the
corresponding coefficient sequences can be readily found.

We begin by noting that since $\{\phi_{\ell}(t)\}$ generate an
orthonormal basis for $\A$, $x(t)$ is uniquely determined by the
$N$ sequences of samples
\begin{equation}
c_{\ell}[n]=\inner{\phi_{\ell}(t-nT)}{x(t)}=r_\ell(nT),
\end{equation}
where $r_\ell(t)$ is the convolution
$r_\ell(t)=\phi_\ell(-t)*x(t)$. Indeed, orthonormality of
$\{\phi_{\ell}(t)\}$ immediately implies that
\begin{equation}
\label{eq:xtc} x(t)=\sum_{\ell=1}^{N} \sum_{n \in \ZZ}
c_\ell[n]\phi_\ell(t-nT).
\end{equation}
Therefore, constraining $x(t)$ is equivalent to imposing
restrictions on the expansion coefficients $c_{\ell}[n]$. Taking
the inner products on both sides of (\ref{eq:decoma}) with respect
to $\phi_{r}(t-mT)$ leads to
\begin{eqnarray}
\label{eq:constc2} c_r[m] & = & \sum_{\ell=1}^{2N} \sum_{n \in
\ZZ}
\gamma_\ell[n]\inner{\phi_{r}(t-mT)}{d_\ell(t-nT)} \nonumber \\
& = & \sum_{\ell=1}^{2N} \sum_{n \in \ZZ}
\gamma_\ell[n]a_{r\ell}[m-n],
\end{eqnarray}
where $a_{r\ell}[n]=\inner{\phi_{r}(t-nT)}{d_\ell(t)}$. In the
Fourier domain, (\ref{eq:constc2}) can be written as
\begin{equation}
\label{eq:constcf} C_r\ejo=\sum_{\ell=1}^{2N} \Gamma_\ell\ejo
A_{r\ell}\ejo,\quad 1 \leq r \leq N.
\end{equation}
Thus, instead of finding $\gamma_{\ell}[n]$ satisfying the
constraints in (\ref{eq:loa}) we can alternatively seek the
smallest number of functions $\Gamma_{\ell}\ejo$ that satisfy
(\ref{eq:constcf}).

To simplify (\ref{eq:constcf}) we use the definition (\ref{eq:dl})
of $d_{\ell}(t)$. Since \sloppy
$\inner{\phi_{r}(t-nT)}{\phi_\ell(t)}=\delta_{r\ell}\delta_{n0}$
and the Fourier transform of
$\inner{\phi_{r}(t-nT)}{\psi_\ell(t)}$ is equal to
$R_{\phi_r\psi_\ell}\ejo$, (\ref{eq:constcf}) can be written as
\begin{equation}
\label{eq:constcf2} C_r\ejo=\Gamma_r\ejo +\sum_{\ell=N+1}^{2N}
\Gamma_\ell\ejo R_{\phi_r\psi_\ell}\ejo.
\end{equation}
Denoting by $\bc\ejo,\bgam\ejo$ the vectors with elements
$C_\ell\ejo,\Gamma_\ell\ejo$ respectively, we can express
(\ref{eq:constcf2}) as
\begin{equation}
\label{eq:constf} \bc\ejo= \left[
\begin{array}{ll}
\bbi & \bbm_{\phi\psi}\ejo
\end{array}
\right] \bgam\ejo,
\end{equation}
where $\bbm_{\phi\psi}(e^{j\omega})$ is the sampled cross
correlation matrix
\begin{equation}
\label{eq:M2}
\bbm_{\phi\psi}(e^{j\omega})=\left[\begin{array}{ccc}
R_{\phi_1\psi_1}\ejo & \ldots &
R_{\phi_1\psi_N}\ejo\\
\vdots & \vdots &  \vdots \\
R_{\phi_N\psi_1}\ejo & \ldots & R_{\phi_N\psi_N}\ejo
\end{array}
\right],
\end{equation}
with $R_{\phi\psi}$ defined by (\ref{eq:R}). Our sparse recovery
problem (\ref{eq:loa}) is therefore equivalent to
\begin{equation}
\label{eq:l1f}
\begin{array}{ll}
\min_{\bgam}& \|\bgam\ejo\|_{2,0} \\
\sto & \bc\ejo=\left[
\begin{array}{ll}
\bbi & \bbm_{\phi\psi}\ejo
\end{array}
\right] \bgam\ejo.
\end{array}
\end{equation}

Problem (\ref{eq:l1f}) resembles the multiple measurement vector
(MMV) problem, in which the goal is to jointly decompose $m$
vectors $\bx_i,1 \leq i \leq m$ in a dictionary $\bbd$
\cite{Chen,Cotter,EM082,ER09}. In the next section we review the
MMV model and a recently developed generalization to the case in
which it is desirable to jointly decompose infinitely many vectors
$\bx_i$ in terms of a given dictionary $\bbd$. This extension is
referred to as the infinite measurement model (IMV) \cite{ME08}.
In Section~\ref{sec:auc} we show how these ideas can be used to
solve (\ref{eq:l1f}).

As we will show, the ability to sparsely decompose a set of
signals in the IMV and MMV settings depends on the properties of
the corresponding dictionary. In our formulation (\ref{eq:l1f}),
the dictionary is given by
\begin{equation}
\label{eq:dom} \bbd\ejo=\left[
\begin{array}{ll}
\bbi & \bbm_{\phi\psi}\ejo
\end{array}
\right].
\end{equation}
The next proposition establishes some properties of $\bbd\ejo$
that will be used in Section~\ref{sec:auc} in order to solve
(\ref{eq:l1f}).
\begin{proposition}
\label{prop:formm} Let $\{\phi_{\ell}(t-nT),\psi_{\ell}(t-nT),n
\in \ZZ,1 \leq \ell \leq N\}$ denote two orthonormal bases for a
SI space $\A$. Let $\bbm_{\phi\psi}\ejo$ denote
 the cross-correlation matrix defined by (\ref{eq:M2}), and let
 $\mu(\Phi,\Psi)$,$\mu(\bphi,\bpsi)$  be the analog and discrete coherence measures defined by
 (\ref{eq:mu}), (\ref{eq:mud}).
 Then, for each $\omega$:
 \begin{enumerate}
 \item $\bbm_{\phi\psi}\ejo$ is a unitary matrix;
 \item $\mu(\bbi,\bbm_{\phi\psi}\ejo) \leq \mu(\Phi,\Psi)$.
 \end{enumerate}
 \end{proposition}
\begin{proof}
See Appendix~\ref{app:formm}.
\end{proof}

\subsection{MMV and IMV Models}
\label{sec:imv}

The basic results of \cite{EB02,GN03,DE03} on expansions in
dictionaries consisting of two orthonormal bases can be
generalized to the MMV problem in which we would like to jointly
decompose $m$ vectors $\bx_i,1 \leq i \leq m$ in a dictionary
$\bbd$. Denoting by $\bbx$ the matrix with columns $\bx_i$, our
goal is to seek a matrix $\bgaml$ with columns $\bgam_i$ such that
$\bbx=\bbd\bgaml$ and $\bgaml$ has as few non-zero rows as
possible. In this model, not only is each representation vector
$\bgam_i$ sparse, but in addition the vectors share a joint
sparsity pattern. The results in \cite{Chen,Cotter,EM082}
establish that under the same conditions as
Proposition~\ref{prop:unique}, the unique $\bgaml$ can be found by
solving an extension of the $\ell_1$ program:
\begin{equation}
\label{eq:l1mmv} \min_{\bgaml}\|\bs(\bgaml)\|_1 \quad \sto
\bbx=\bbd\bgaml.
\end{equation}
Here $\bs(\bgaml)$ is a vector whose $\ell$th element is equal to
$\|\bgaml^{\ell}\|$ where $\bgaml^{\ell}$ is the $\ell$th row of
$\bgaml$, and the norm is an arbitrary vector norm. When $\bgaml$
is equal to a single vector $\bgam$,
$\|\bgaml^{\ell}\|=|\gamma_\ell|$ for any choice of norm and
(\ref{eq:l1mmv}) reduces to the standard $\ell_1$ optimization
problem (\ref{eq:l1d}).
\begin{proposition}
\label{prop:mmv} Let $\bbx$ be an $N \times m$ matrix with columns
$\bx_i,1 \leq i \leq m$ that have a joint sparse representation in
the dictionary $\bbd=[\bphi\,\,\bpsi]$ consisting of two
orthonormal bases, so that $\bbx=\bbd\bgaml$ with
$\|\bs(\bgaml)\|_0=k$. If $k<1/\mu(\bphi,\bpsi)$ where
$\mu(\bphi,\bpsi)=\max_{\ell,r}|\bphil^H_{\ell}\bpsil_r|$, then
this representation is unique. Furthermore, if
\begin{equation}
k < \frac{\sqrt{2}-0.5}{\mu(\bphi,\bpsi)},
\end{equation}
then the unique sparse representation can be found by solving
(\ref{eq:l1mmv}) with any vector norm.
\end{proposition}

The MMV model has been recently generalized to the IMV case in
which there are infinitely many vectors $\bx$ of length $N$, and
infinitely many coefficient vectors $\bgam$:
\begin{equation}
\label{eq:imv} \bx(\lambda)=\bbd\bgam(\lambda),\quad \lambda \in
\Lambda,
\end{equation}
where $\Lambda$ is some set whose cardinality can be infinite. In
particular, $\Lambda$ may be uncountable, such as the set of
frequencies $\omega \in (-\pi,\pi]$. The $k$-sparse IMV model
assumes that the vectors $\{\bgam(\lambda)\}$, which we denote for
brevity by $\bgam(\Lambda)$, share a joint sparsity pattern, so
that the non-zero elements  are all supported on a fixed location
set of size $k$ \cite{ME08}. This model was first introduced in
\cite{ME07} in the context of blind sampling of multiband signals,
and later analyzed in more detail in \cite{ME08}.

A major difficulty with the IMV model is that it is not clear in
practice how to determine the entire solution set $\bgam(\Lambda)$
since there are infinitely many equations to solve. Thus, using an
$\ell_1$ optimization, or a greedy approach, are not immediately
relevant here. In \cite{ME08} it was shown that (\ref{eq:imv}) can
be converted to a finite MMV without loosing any information by a
set of operations that are grouped under a block refereed to as
the continuous-to-finite (CTF) block. The essential idea is to
first recover the support of $\bgam(\Lambda)$, namely the non-zero
location set, by solving a finite MMV. We then reconstruct
$\bgam(\Lambda)$ from the data $\bx(\Lambda)$ and the knowledge of
the support, which we denote by $S$. The reason for this
separation is that once $S$ is known, the linear relation of
(\ref{eq:imv}) becomes invertible when the coherence is low
enough.

To see this, let $\bbd_S$ denote the matrix containing the subset
of the columns of $\bbd$ whose indices belong to $S$. The system
of (\ref{eq:imv}) can then be written as
\begin{equation}\label{yAxS}
\bx(\lambda) = \bbd_S \bgam^S(\lambda),\quad\lambda\in\Lambda,
\end{equation}
where the superscript $\bgam^S(\lambda)$ is the vector that
consists of the entries of $\bgam(\lambda)$ in the locations $S$.
 Since $\bgam(\Lambda)$ is $k$-sparse, $|S|\leq k$. In
addition, from Proposition~\ref{prop:kr} it follows that if
$\mu(\bphi,\bpsi)<1/k$ then every $k$ columns of $\bbd$ are
linearly independent. Therefore $\bbd_S$ consists of linearly
independent columns implying that $\bbd_S^\dag\bbd_S=\bbi$, where
$\bbd_S^\dag = \left(\bbd_S^H\bbd_S\right)^{-1}\bbd_S^H$ is the
Moore-Penrose pseudo-inverse of $\bbd_S$. Multiplying (\ref{yAxS})
by $\bbd_S^\dag$ on the left gives
\begin{equation}\label{Reconstruct1}
\bgam^S(\lambda) = \bbd_S^\dag\bx(\lambda),\quad\lambda\in\Lambda.
\end{equation}
The elements in $\bgam(\lambda)$ not supported on $S$ are all
zero. Therefore (\ref{Reconstruct1}) allows for exact recovery of
$\bgam(\Lambda)$ once the finite set $S$ is correctly identified.

In order to determine $S$ by solving a finite-dimensional problem
we exploit the fact that $\Span(\bx(\Lambda))$ is finite, since
$\bx(\lambda)$ is of length $N$. Therefore, $\Span(\bx(\Lambda))$
has dimension at most $N$. In addition, it is shown in \cite{ME08}
that if there exists a solution set $\bgam(\Lambda)$ with sparsity
$k$, and the matrix $\bbd$ has Kruskal rank $\sigma(\bbd) \geq
2k$, then every finite collection of vectors spanning the subspace
$\Span(\bx(\Lambda))$ contains sufficient information to recover
$S$ exactly. Therefore, to find $S$ all we need is to construct a
matrix $\bbv$ whose range space is equal to $\Span(\bx(\Lambda))$.
We are then guaranteed that the linear system
\begin{equation}\label{VAU}
\bbv=\bbd\bbu
\end{equation}
has a unique $k$-sparse solution $\bbu$ whose row support is equal
$S$.
 This result allows to avoid the infinite
structure of (\ref{eq:imv}) and to concentrate on finding the
finite set $S$ by solving the single MMV system of (\ref{VAU}).
The solution can be determined using an $\ell_1$ relaxation of the
form (\ref{eq:l1mmv}) with $\bbv$ replacing $\bbx$, as long as the
conditions of Proposition~\ref{prop:mmv} hold, namely the
coherence is small enough with respect to the sparsity.

In practice, a matrix $\bbv$ with column span equal to
$\Span(\bx(\Lambda))$ can be constructed by first forming the
matrix $\bbq
=\int_{\lambda\in\Lambda}\bx(\lambda)\bx^H(\lambda)d\lambda$,
assuming that the integral exists. Every $\bbv$ satisfying
$\bbq=\bbv\bbv^H$ will then have a column span equal to
$\Span(\bx(\Lambda))$ \cite{ME08}. In particular, the columns of
$\bbv$ can be chosen as the eigenvectors of $\bbq$ multiplied by
the square-root of the corresponding eigenvalues.

We summarize the steps enabling a finite-dimensional solution to
the IMV problem in the following theorem.
\begin{theorem}\label{ThKey}
Consider the system of equations (\ref{eq:imv}) where
$\bbd=[\bphi\,\,\bpsi]$ is a dictionary consisting of two
orthonormal bases with coherence
$\mu(\bphi,\bpsi)=\max_{\ell,r}|\bphil^H_{\ell}\bpsil_r|$. Suppose
(\ref{eq:imv}) has a $k$-sparse solution set $\bgam(\Lambda)$ with
support set $S$. If the Kruskal rank $\sigma(\bbd) \geq 2k$, then
$\bgam(\Lambda)$ is unique. In addition, let $\bbv$ be a matrix
whose column-space is equal to $\Span(\bx(\Lambda))$. Then, the
linear system $\bbv=\bbd\bbu$ has a unique $k$-sparse solution
$\bbu$ whose row support is equal to $S$. Denoting by $\bbd_S$ the
columns of $\bbd$ whose indices belong to $S$, the non-zero
elements $\bgam^S(\lambda)$ are given by
$\bgam^S(\lambda)=\bbd_S^\dagger\bx(\lambda)$.
 Finally, if
\begin{equation}
k < \frac{\sqrt{2}-0.5}{\mu(\bphi,\bpsi)},
\end{equation}
then $\sigma(\bbd) \geq 2k$ and the unique sparse $\bbu$ can be
found by solving (\ref{eq:l1mmv}) with any vector norm.
\end{theorem}

\subsection{Analog Dictionaries}
\label{sec:auc}

In Section~\ref{sec:ar} we showed that the analog decomposition
problem (\ref{eq:loa}) is equivalent to (\ref{eq:l1f}). The later
is very similar to the IMV problem (\ref{eq:imv}). Indeed, we seek
a continuous set of vectors $\bgam$ with joint sparsity that have
the smallest number of non-zero rows, and satisfy an infinite set
of linear equations. However, in contrast to (\ref{eq:imv}), the
matrix in (\ref{eq:l1f}) depends on $\omega$. Therefore,
Theorem~\ref{ThKey} cannot be applied since it is not clear what
matrix figures in the finite MMV representation.
 Nonetheless, the essential idea of
separating the support recovery from that of the actual values of
$\bgam \ejo$ is still valid. In particular, we can solve
(\ref{eq:l1f}) by first determining the support set of $\bgam
\ejo$. Once the support is known, we have that
\begin{equation}\label{eq:bgams}
\bgam^S \ejo =( \bbd_S^H\ejo\bbd_S\ejo)^{-1}\bbd_S^H\ejo \bc\ejo,
\end{equation}
where $\bbd\ejo$ is defined by (\ref{eq:dom}). The inverse in
(\ref{eq:bgams}) exists if $\mu(\bbi,\bbm_{\phi\psi}\ejo)$ is
smaller than $1/k$. From Proposition~\ref{prop:formm}, it is
sufficient to require that $\mu(\Phi,\Psi)<1/k$.

To find the support set $S$ we distinguish between two different
cases:
\begin{enumerate}
\item The constant case in which $\bbm_{\phi\psi}\ejo$ of
(\ref{eq:M2}) can be written as
\begin{equation}
\label{eq:rsc} \bbm_{\phi\psi}\ejo=\bba \bbz\ejo.
\end{equation}
Here $\bba$ is a fixed matrix independent of $\omega$, and
$\bbz\ejo$ is an invertible diagonal matrix with diagonal elements
$Z_{\ell}\ejo$; the columns of $\bba$ are normalized such that
$\ess\sup |Z_{\ell}\ejo|=1$ for all $\ell$.\item The rich case in
which the support of every subset of $\bgam \ejo$ of a given size
$M$, is equal to the support $S$ of the entire set.
\end{enumerate}
The first case involves a condition on the dictionary. The second
allows for arbitrary dictionaries, but imposes a constraint on the
expansion sequences. This restriction is quite mild, and satisfied
for a large class of dictionaries and signals. In both cases we
show that the support can be found by solving a finite-dimensional
optimization problem.

\textit{Constant case:} We begin by treating the setting in which
the sampled cross correlation matrix can be written as in
(\ref{eq:rsc}). For example, consider the case in which $\A$ is
the space of real signals bandlimited to $\pi N/T$, as in
Section~\ref{sec:example}. Then $\phi_{\ell}(t),\psi_{\ell}(t)$
defined by (\ref{eq:psie}), (\ref{eq:phie}) satisfy (\ref{eq:rsc})
(for $\omega \geq 0$) with $\bba=(1/\sqrt{N}) \overline{\bbf}$,
where $\bbf$ denotes the $N \times N$ Fourier matrix and
$Z_{\ell}\ejo=\exp\{j\omega(\ell-1)/N\}$.

The unitarity of $\bbm_{\phi\psi}\ejo$, which follows from
Proposition~\ref{prop:formm}, implies that
$\bba=\bbm_{\phi\psi}\ejo\bbz^{-1}\ejo$ must be unitary as well.
Indeed, for all $\omega$, we have
\begin{equation}
\label{eq:zh} \bba^H\bba=(\bbz\ejo \bbz^H \ejo)^{-1}.
\end{equation}
Therefore, $|Z_{\ell}\ejo|$ is independent of $\omega$. Since
$\max_\omega |Z_{\ell}\ejo|=1$, we conclude that
$|Z_{\ell}\ejo|=1$ for all $\omega$ so that $\bbz\ejo \bbz^H
\ejo=\bbi$, which together with (\ref{eq:zh}) proves the unitarity
of $\bba$.

To obtain a correlation structure of the form (\ref{eq:rsc}) we
may start with a given orthonormal basis $\{\psi_\ell(t-nT)\}$,
and then create another orthonormal basis $\{\phi_\ell(t-nT)\}$ by
choosing
\begin{equation}
\phi_{\ell}(t)=\sum_{r=1}^N \sum_{n \in \ZZ}
a_r^{\ell}[n]\psi_r(t-nT).
\end{equation}
Here $a_r^{\ell}[n]$ is any set of sequences for which
$A_r^{\ell}\ejo=[\bba]_{\ell r}Z_r\ejo$ with $\bba$ an arbitrary
unitary matrix, and $\bbz$ is an arbitrary diagonal unitary
matrix. This is a direct consequence of the proof of
Proposition~\ref{prop:formm}.

Under the condition (\ref{eq:rsc}) we now show that we can
convert (\ref{eq:l1f}) to a finite MMV problem. Indeed, let the
first $N$ elements of $\bgam\ejo$ be denoted by $\ba \ejo$ and the
remaining $N$ elements by $\bb \ejo$. Then
 (\ref{eq:l1f}) becomes
\begin{equation}
\label{eq:l1ff}
\begin{array}{ll}
\min_{\ba,\bd} & \|\ba\ejo\|_{2,0}+\|\bd\ejo\|_{2,0}
\\
\sto & \bc\ejo=\left[
\begin{array}{ll}
\bbi & \bba
\end{array}
\right] \left[
\begin{array}{l}
\ba\ejo \\ \bd\ejo
\end{array}
\right],
\end{array}
\end{equation}
where $\bd \ejo=\bbz\ejo\bb\ejo$, and we used the fact that since
$\bbz\ejo$ is diagonal and invertible,
$\|\bb\ejo\|_{2,0}=\|\bd\ejo\|_{2,0}$ so that the two vector
sequences have the same sparsity. Problem (\ref{eq:l1ff}) has the
required IMV form. It can be solved by first finding the sparsest
matrix $\bbu$ that satisfies $\bbc=[\bbi\,\,\bba]\bbu$ where the
columns of $\bbc$ form a basis for the span of $\{\bc\ejo,-\pi
\leq \omega \leq \pi\}$. As we have seen, a basis can be
determined in frequency by first forming the correlation matrix
\begin{equation}
\bbq =\int_{-\pi}^\pi \bc\ejo \bc^H\ejo d\omega.
\end{equation}
Alternatively, we can find a basis in time by creating
\begin{equation}
\bbq' =\sum_{n = -\infty}^\infty \bc[n]\bc^H[n].
\end{equation}
The basis can then be chosen as the eigenvectors corresponding to
nonzero eigenvalues of $\bbq$ or $\bbq'$, which we denote by
$\bbc$.
 To find $\bbu$
we consider the convex program
\begin{equation} \label{eq:asmmv}
\min_{\bbu}\|\bs(\bbu)\|_1 \quad \sto \bbc=\left[
\begin{array}{ll}
\bbi & \bba
\end{array}
\right] \bbu.
\end{equation}

Let $S$ denote the rows in $\bbu$ that are not identically zero
and let $\bgam^S[n]$ be the corresponding sequences
$\gamma_{\ell}[n],\ell \in S$. Then
\begin{equation}
\bgam^S\ejo=\left[
\begin{array}{c}
\bbi \\
\bbz_{S'}^{-1}\ejo
\end{array}
\right] (\bbd_S^H\bbd_S)^{-1}\bbd_S^H\bc\ejo,
\end{equation}
where $\bbd=[\bbi\,\,\bba]$, and $S'$ denotes the rows in $S$
between $N+1$ and $2N$. The remaining sequences
$\gamma_{\ell},\ell \notin S$ are identically zero.
Proposition~\ref{prop:mmv} provides conditions under which
(\ref{eq:asmmv}) will find the sparsest representation in terms of
the coherence $\mu(\bbi,\bba)$ (where we rely on the fact that
$\bba$ is unitary). Since $|Z_{\ell}\ejo|=1$, we have that
$|[\bba\bbz\ejo]_{ij}|=|\bba_{ij}|$ and
$\mu(\bbi,\bba)=\mu(\Phi,\Psi)$.

We summarize our results on analog sparse decompositions in the
following theorem.
\begin{theorem}
\label{thm:as}  Let $\{\phi_{\ell}(t),1 \leq \ell \leq N\}$ and
$\{\psi_{\ell}(t),1 \leq \ell \leq N\}$ denote two orthonormal
generators of a SI subspace $\A$ of $L_2$ with coherence
$\mu(\Phi,\Psi)$. Let $x(t)$ be a signal in $\A$ and suppose there
exists sequences $a_{\ell}[n],b_{\ell}[n]$ such that
\begin{equation}
x(t)=\sum_{\ell=1}^{N} \sum_{n \in \ZZ}
(a_\ell[n]\phi_\ell(t-nT)+b_\ell[n]\psi_\ell(t-nT))
\end{equation}
with $k=\|\ba\|_{2,0}+\|\bb\|_{2,0}$ satisfying $k <
(\sqrt{2}-0.5)/\mu(\Phi,\Psi)$. Let $\bbm_{\phi\psi}\ejo$ be the
cross-correlation matrix defined by (\ref{eq:M2}) and suppose that
it can be written as $\bbm_{\phi\psi}\ejo=\bba\bbz\ejo$, where
$\bba$ is unitary and $\bbz\ejo$ is a diagonal unitary matrix.
Then, the sequences $a_{\ell}[n]$ and $b_{\ell}[n]$ can be found
by solving
\begin{equation}
\label{eq:asmmvt} \begin{array}{ll}
 \min_{\bgaml_1,\bgaml_2} &
\|\bs(\bgaml_1)\|_1+\|\bs(\bgaml_2)\|_1 \\
\sto & \bbc=\left[
\begin{array}{ll}
\bbi & \bba
\end{array}
\right] \left[
\begin{array}{ll}
\bgaml_1 \\ \bgaml_2
\end{array}
\right].
\end{array}
\end{equation}
Here $\bbc$ is chosen such that its columns form a basis for the
range of $\{\bc\ejo,\omega \in (-\pi,\pi]\}$ where the $\ell$th
component of $\bc\ejo$ is the Fourier transform at frequency
$\omega$ of $c_{\ell}[n]=\inner{\phi_{\ell}(t-nT)}{x(t)}$, and
$\bs(\bgaml_i)$ is a vector whose $\ell$th element is equal to
$\|\bgaml_i^{\ell}\|$ where the norm is arbitrary. Let $S_1,S_2$
denote the rows of $\bgaml_1,\bgaml_2$ that are not identically
equal $0$, and define $\bbd_S=[\bbi_{S_1}\,\,\bba_{S_2}]$. Then
the non-zero sequences $a_{\ell}[n],b_{\ell}[n],\ell \in S$ are
given in the Fourier domain by
\begin{equation}
\label{eq:asthm} \left[
\begin{array}{c}
\ba_S\ejo \\
\bb_S\ejo
\end{array}
\right] =\left[
\begin{array}{c}
\bbi \\
\bbz^{-1}_{S_2}\ejo
\end{array}
\right] (\bbd_S^H\bbd_S)^{-1}\bbd_S^H\bc\ejo.
\end{equation}
\end{theorem}

In Theorem~\ref{thm:as} the sparse decomposition is determined
from the samples $c_{\ell}[n]=\inner{\phi_{\ell}(t-nT)}{x(t)}$.
However, the theorem also holds when $c_{\ell}[n]$ is replaced by
any sequence of samples $\inner{h_{\ell}(t-nT)}{x(t)}$ with
$h_{\ell}(t)$ being an orthonormal basis for $\A$ such that both
$\bbm_{h\phi}\ejo$ and $\bbm_{h\psi}\ejo$ are constant up to a
diagonal matrix:
\begin{equation}
\bbm_{h\phi}\ejo=\bba_1 \bbz_1\ejo,\quad
\bbm_{h\psi}\ejo=\bba_2\bbz_2\ejo.
\end{equation}
In this case the matrix $[\bbi\,\,\bba]$ in (\ref{eq:asmmvt})
should be replaced by the matrix $[\bba_1\,\,\bba_2]$. Once we
find the sparsity set $S$, the sequences that are not zero can be
found as in (\ref{eq:asthm}) with the identity in the first matrix
replaced by the appropriate rows of $\bbz_1^{-1}\ejo$.

\textit{Rich case:} We next consider the case of an arbitrary
$\bbd\ejo$, and impose a condition on the sequences
$\gamma_\ell[n]$. Specifically, we assume that there exists a
finite number $M$ such that the support set of
$\{\bgam(e^{j\omega_i}),|i|=M\}$ is equal $S$. In other words, the
joint support of any $M$ vectors $\bgam(e^{j\omega_i})$ is equal
to the support of the entire set. Under this assumption, the
support recovery problem reduces to an MMV model and can therefore
be solved efficiently using MMV techniques. Specifically, we
select a set of $M$ frequencies $\omega_i$, and seek the matrix
$\bgaml$ with columns $\bgam_i$ that is the solution to
\begin{equation}
\label{eq:mmvtw} \begin{array}{ll} \min_{\bgaml}&
\|\bs(\bgaml)\|_1
 \\
\sto & \bc(e^{j\omega_i})= \left[
\begin{array}{ll}
\bbi & \bbm_{\phi\psi}(e^{j\omega_i})
\end{array}
\right] \bgam_i,\quad 1 \leq i \leq M.
\end{array}
\end{equation}

 If
we choose $\bs(\bgaml)$ as the $\ell_1$ norm, then
(\ref{eq:mmvtw}) is equivalent to $M$ separate problems, each of
the form
\begin{equation}
\label{eq:mmvtw2} \min_{\bgam}\|\bgam\|_1 \quad \sto \bc= \left[
\begin{array}{ll}
\bbi & \bbu
\end{array}
\right] \bgam,
\end{equation}
were $\bc=\bc(e^{j\omega_i})$ and
$\bbu=\bbm_{\phi\psi}(e^{j\omega_i})$ is a unitary matrix (see
Proposition~\ref{prop:formm}). From Proposition~\ref{prop:unique},
the correct sparsity pattern will be recovered if $\mu(\bbi,\bbu)$
is low enough, which due to Proposition~\ref{prop:formm} can be
guaranteed by upper bounding $\mu(\Phi,\Psi)$.

 In some cases, even one
frequency $\omega_i$ may be sufficient in order to determine the
correct sparsity pattern; this happens when the support of
$\bgam(e^{j\omega_i})$ is equal to the support of the entire set
of sequences $\bgam(e^{j\omega})$.  In practice, we can solve for
an increasing number of frequencies, with the hope of recovering
the entire support in a finite number of steps. Although we can
always construct a set of signals whose joint support cannot be
detected in a finite number of steps, this class of signals is
small. Therefore, if the sequences are generated at random, then
with high probability choosing a finite number of frequencies will
be sufficient to recover the entire support set.

\section{Extension to Arbitrary Dictionaries} \label{sec:frame}

Until now we discussed the case of a dictionary comprised of two
orthonormal bases. The theory we developed can easily be extended
to treat the case of an arbitrary dictionary comprised of
sequences $d_{\ell}(t)$ that form a frame (\ref{eq:frame}) for
$\A$. These results follow from combining the approach of the
previous section with the corresponding statements in the discrete
setting developed in \cite{GN03,DE03,T04}.

Specifically, suppose we would like to decompose a vector $\bx \in
\RR^N$ in terms of a dictionary $\bbd$ with columns $\bd_{\ell}$
using as few vectors as possible. This corresponds to solving
\begin{equation}
\label{eq:loai} \min_{\bgam}\|\bgam\|_0 \quad \sto  \bx=\bbd\bgam.
\end{equation}
Since (\ref{eq:loai}) has combinatorial complexity, we would like
to replace it with a computationally efficient algorithm. If
$\bbd$ has low coherence, where in this case the coherence is
defined by
\begin{equation}
\label{eq:muD} \mu(\bbd)= \max_{\ell \neq r}
\frac{|\bd_\ell^H\bd_r|}{\|\bd_{\ell}\|\|\bd_{r}\|},
\end{equation}
then we can determine the sparsest solution $\bgam$ by solving the
$\ell_1$ problem
\begin{equation}
\label{eq:l1a2} \min_{\bgam}\|\bgam\|_1 \quad \sto  \bx=\bbd\bgam.
\end{equation}
The coherence of a dictionary measures the similarity between its
elements and is equal to $0$ only if the dictionary consists of
orthonormal vectors. A general lower bound on the coherence of a
matrix $\bbd$ of size $N \times m$ is \cite{T04} $\mu(\bbd) \geq
[(m-N)/(N(m-1))]^{1/2}$. The same results hold true for the
corresponding MMV model, and are incorporated in the following
proposition \cite{DE03,GN03,T04,Chen}:
\begin{proposition}
\label{prop:l1da} Let $\bbd$ be an arbitrary dictionary with
coherence $\mu(\bbd)$ given by (\ref{eq:muD}). Then the Kruskal
rank satisfies $\sigma(\bbd)>1/\mu(\bbd)-1$. Furthermore, if there
exists a choice of coefficients $\bgaml$ such that
$\bbx=\bbd\bgaml$ and
\begin{equation}
\label{eq:gamca} \|\bs(\bgaml)\|_0 <\frac{1}{2}\bl
1+\frac{1}{\mu(\bbd)}\br,
\end{equation}
then the unique sparse representation can be found by solving
(\ref{eq:l1mmv}).
\end{proposition}

We now apply Proposition~\ref{prop:l1da} to the analog design
problem. Suppose we have a signal $x(t)$ that lies in a SI space
$\A$, and let $\{d_{\ell}(t-nT),1 \leq \ell \leq m\}$ denote an
arbitrary frame for $\A$ with $m>N$. As an example, consider the
space $\A$ of real signals bandlimited to $(-\pi N/T,\pi N/T]$,
which was introduced in Section~\ref{sec:example}. As we have
seen, this space can be generated by the $N$ functions
\begin{equation}
\phi_{\ell}(t)=\frac{1}{\sqrt{T'}}\sinc((t-(\ell-1) T')/T'),\quad
1 \leq \ell \leq N,
\end{equation}
with $T'=T/N$. Suppose now that we define the functions
\begin{equation}
\label{eq:psitf}
\tilde{\phi}_{\ell}(t)=\frac{1}{\sqrt{\tilde{T}}}\sinc((t-(\ell-1)
\tilde{T})/\tilde{T}),\quad 1 \leq \ell \leq m,
\end{equation}
where $\tilde{T}=T/m$ and $m>N$. Using similar reasoning as that
used to establish the basis properties of the generators
(\ref{eq:psit}), it is easy to see that
$\{\tilde{\phi}_{\ell}(t)\}$ constitute an orthonormal basis for
the space of signals bandlimited to $(-\pi m/T,\pi m/T]$ which is
larger than $\A$. Filtering each one of the basis signals with a
(scaled) LPF with cut-off $\pi /T'$ will result in a redundant set
of functions
\begin{equation}
\label{eq:da} d_{\ell}(t)=\frac{1}{\sqrt{T'}}\sinc((t-(\ell-1)
\tilde{T})/T'),\quad 1 \leq \ell \leq m,
\end{equation}
that form a frame for $\A$ \cite{O02,MEp08}.

Our goal is to represent a signal $x(t)$ in $\A$ using as few
sequences $d_{\ell}(t)$ as possible. More specifically, our
problem is to choose the vector sequence $\bgam[n]$ such that
\begin{equation}
\label{eq:xd} x(t)=\sum_{\ell=1}^{m} \sum_{n \in \ZZ}
\gamma_\ell[n]d_\ell(t-nT),
\end{equation}
and $\|\bgam\|_{2,0}$ is minimized.

To derive an infinite-dimensional alternative to (\ref{eq:l1a2})
let $\{h_{\ell}(t)\}$ generate a basis for $\A$. Then $x(t)$ is
uniquely determined by the $N$ sampling sequences
\begin{equation}
c_{\ell}[n]=\inner{h_{\ell}(t-nT)}{x(t)}=r_{\ell}(nT),
\end{equation}
where $r_{\ell}(t)$ is the convolution $r_{\ell}(t)=h(-t)*x(t)$.
 Therefore, $x(t)$ satisfies (\ref{eq:xd}) only if
\begin{equation}
\label{eq:constca} c_r[m]=\sum_{\ell=1}^{m} \sum_{n \in \ZZ}
\gamma_\ell[n]a_{r\ell}[n],
\end{equation}
where $a_{r\ell}[n]=\inner{h_{r}(t-nT)}{d_\ell(t)}$. In the
Fourier domain (\ref{eq:constca}) becomes
\begin{equation}
\label{eq:constcfa} C_r\ejo=\sum_{\ell=1}^{m} \Gamma_\ell\ejo
A_{r\ell}\ejo=\sum_{\ell=1}^{m} \Gamma_\ell\ejo R_{h_rd_\ell}\ejo.
\end{equation}
Denoting by $\bc\ejo,\bgam\ejo$ the vectors with elements
$C_\ell\ejo,\Gamma_\ell\ejo$ respectively we can write
(\ref{eq:constcfa}) as
\begin{equation}
\label{eq:finalat} \bc\ejo=\bbm_{hd}\ejo\bgam\ejo.
\end{equation}
Therefore, our problem is to find the sparsest set of $\bgam\ejo$
that satisfies (\ref{eq:finalat}).

In order to solve the sparse decomposition problem we first treat
the case in which $\{h_{\ell}(t)\}$ are chosen such that
\begin{equation}
\label{eq:rsc2} \bbm_{hd}\ejo=\bbw\ejo\bba \bbz\ejo,
\end{equation}
where $\bba$ is a fixed matrix independent of $\omega$,
$\bbz\ejo$ is an invertible diagonal matrix with diagonal elements
$Z_{\ell}\ejo$ satisfying  $\ess\sup |Z_{\ell}\ejo|=1$, and
$\bbw\ejo$ is an arbitrary invertible matrix. Going back to the
bandlimited frame (\ref{eq:da}) it can be easily seen that with
$h_{\ell}(t)=\phi_{\ell}(t)$, (\ref{eq:rsc2}) is satisfied.
Indeed,
\begin{eqnarray}
\lefteqn{\overline{H}_{\ell}(\omega)D_{r}(\omega)= } \nonumber \\
&& \hspace*{-0.2in} \left\{
\begin{array}{ll}
\frac{T}{N} e^{j \omega (\ell-1) T/N}e^{-j \omega (r-1) T/m}, & \omega \in (-\pi N/T,\pi N/T]; \\
0, & \mbox{otherwise}.
\end{array}
\right.
\end{eqnarray}
Therefore,
\begin{equation}
R_{h_\ell d_r} \ejo = e^{j \omega (\ell-1) /N}e^{-j \omega (r-1)
/m} f(\ell,r),
\end{equation}
where $f(\ell,r)$ is a function only of the indices $\ell,r$ and
not the frequency $\omega$. Choosing $Z_{r} \ejo=e^{-j \omega
(r-1)/m}$ and $\bbw \ejo $ as a diagonal matrix with diagonal
elements $W_{\ell} \ejo=e^{j \omega (\ell-1)/N}$ leads to the
representation (\ref{eq:rsc2}).

When $\bbm_{hd}\ejo$ has the form (\ref{eq:rsc2}), the system of
equations (\ref{eq:finalat}) becomes
\begin{equation}
\label{eq:finala} \bd\ejo=\bba\bbz\ejo\bgam\ejo=\bba\ba\ejo,
\end{equation}
where we denoted $\bd\ejo=\bbw^{-1}\ejo\bc\ejo$,
$\ba\ejo=\bbz\ejo\bgam\ejo$ and used (\ref{eq:rsc2}). Clearly,
$\|\ba\ejo\|_{2,0}= \|\bgam\ejo\|_{2,0}$
 because $\bbz\ejo$ is invertible and
diagonal. Therefore, the sparse decomposition problem is
equivalent to finding $\ba\ejo$ satisfying (\ref{eq:finala}) and
such that $\|\ba\ejo\|_{2,0}$ is minimized.

 As in the previous section, the sparsest $\ba\ejo$ can be determined by
first converting (\ref{eq:finala}) to a finite MMV problem, in
which we seek the sparsest matrix $\bbu$ that satisfies
$\bbc=\bba\bbu$ where the columns of $\bbc$ form a basis for the
span of $\{\bbw^{-1}\ejo\bc\ejo,-\pi \leq \omega \leq \pi\}$. The
matrix $\bbu$ can be determined by solving the convex problem
\begin{equation}
\label{eq:asmmv2} \min_{\bbu}\|\bs(\bbu)\|_1 \quad \sto
\bbc=\bba\bbu.
\end{equation}
From Proposition~\ref{prop:l1da} it follows that the unique sparse
matrix $\bbu$ can be recovered as long as $\mu(\bba)$ satisfies
(\ref{eq:gamca}). Once we determine the non-zero rows $S$ in
$\bbu$, we can find the non-zero sequences $\bgam^S[n]$ by noting
that from Proposition~\ref{prop:l1da} the columns $\bba_S$ of
$\bba$ corresponding to $S$ are linearly independent. Therefore,
\begin{equation}
\label{eq:finals}
\bgam^S\ejo=\bbz_S^{-1}\ejo(\bba_S^H\bba_S)^{-1}\bba_S^H\bbw^{-1}\ejo\bc\ejo.
\end{equation}

If (\ref{eq:rsc2}) is not satisfied, but instead $\bgam\ejo$ is
rich, so that the support of every $M$ set of vectors (for $M$
different frequencies) is equal to the span of the entire set,
then we can still convert the problem into an MMV. To do this, we
choose $M$ frequency values and seek the set of vectors $\bgam_i,1
\leq i \leq M$ with the sparsest joint support that satisfy
\begin{equation}
\label{eq:finalat2}
\bc(e^{j\omega_i})=\bbm_{hd}(e^{j\omega_i})\bgam_i,\quad 1 \leq i
\leq M.
\end{equation}
Once the support is determined, we can find the non-zero sequences
$\bgam^S[n]$ using (\ref{eq:finals}).

We have outlined a concrete method to find the sparsest
representation of a signal $x(t)$ in $\A$ in terms of an arbitrary
dictionary. In our proposed approach, the reconstruction is
performed with respect to the samples
$c_{\ell}[n]=\inner{h_{\ell}(t-nT)}{x(t)}$. We may alternatively
view our algorithm as a method to reconstruct $x(t)$ from these
samples assuming the knowledge that $x(t)$ has a sparse
decomposition in the given dictionary. Thus, our results can also
be interpreted as a reconstruction method from a given set of
samples, and in that sense complements the results of \cite{E08}.

\section{Conclusion}

In this paper, we extended the recent line of work on generalized
uncertainty principles to the analog domain, by considering sparse
representations in SI bases. We showed that there is a fundamental
limit on the ability to sparsely represent an analog signal in an
infinite-dimensional SI space in two orthonormal bases. The
sparsity bound is similar to that obtained in the
finite-dimensional discrete setting: In both cases the joint
sparsity is limited by the inverse coherence of the bases.
However, while in the finite setting, the coherence is defined as
the maximal absolute inner product between elements from each
basis, in the analog problem the coherence is the maximal absolute
value of the sampled cross-spectrum between the signals.

As in the finite  domain, we can show that the proposed
uncertainty relation is tight by providing a concrete example in
which it is achieved. Our example mimics the finite setting by
considering the class of bandlimited signals as the signal space.
This leads to a Fourier representation that is defined over a
finite, albeit continuous, interval. Within this space we can
achieve the uncertainty limit by considering a bandlimited train
of LPFs. This choice of signal resembles the spike train which is
known to achieve the uncertainty principle in the discrete
setting.

Finally, we treated the problem of sparsely representing an analog
signal in an overcomplete dictionary. Building upon the
uncertainty principle and recent works in the area of compressed
sensing for analog signals, we showed that under certain
conditions on the Fourier domain representation of the dictionary,
the sparsest representation can be found by solving a
finite-dimensional convex optimization problem. The fact that
sparse decompositions can be found by solving a convex
optimization problem has been established in many previous works
in compressed sensing in the finite setting. The additional twist
here is that even though the problem has infinite dimensions, it
can be solved exactly by a finite-dimensional program in many
interesting cases.

In this paper we have focused on analog signals in SI spaces. A
very interesting further line of research is to extend these ideas
and notions to a larger class of analog signals, leading to a
broader notion of analog sparsity and analog compressed sensing.

\section{Acknowledgement}

The author would like to thank Prof. Arie Feuer for carefully
reading a draft of the manuscript and providing many constructive
comments.

\appendices

\section{Proof of Proposition~\ref{prop:orth}}
\label{app:energy}

 To prove the proposition, note that
\begin{eqnarray}
\label{eq:norm} \lefteqn{\int_{-\infty}^\infty |x(t)|^2dt  =
\frac{1}{2\pi}\int_{-\infty}^\infty \left|
X(\omega)\right|^2 d\omega }\nonumber \\
&= &\frac{1}{2\pi}\int_{-\infty}^\infty \left| \sum_{\ell=1}^N
A_{\ell}\ejt \Phi_{\ell}(\omega) \right|^2 d\omega,
\end{eqnarray}
where the last equality follows from (\ref{eq:xeq}). To simplify
(\ref{eq:norm}) we rewrite the integral over the entire real line,
as the sum of integrals over intervals of length $2\pi/T$:
\begin{equation}
\label{eq:integrals} \int_{-\infty}^\infty X(\omega)
d\omega=\int_0^{\frac{2\pi}{T}}\sum_{k=-\infty}^\infty X\bl \omega
- \frac{2\pi}{T}k \br d\omega,
\end{equation}
for all $X(\omega)$. Substituting into (\ref{eq:norm}) and using
the fact that $A_{\ell}\ejt$ is $2\pi/T$-periodic, we obtain
\begin{eqnarray}
\lefteqn{\int_{-\infty}^\infty |x(t)|^2dt =} \nonumber \\
& = &
\frac{1}{2\pi}\int_0^{\frac{2\pi}{T}}\sum_{k=-\infty}^\infty
\left| \sum_{\ell=1}^N A_{\ell}\ejt \Phi_{\ell}\bl \omega -
\frac{2\pi}{T}k \br \right|^2 d\omega \nonumber \\
& = & \frac{T}{2\pi}\int_0^{\frac{2\pi}{T}} \sum_{\ell=1}^N
\sum_{r=1}^N \overline{A_{\ell}}\ejt A_{r}\ejt
R_{\phi_{\ell}\phi_r}\ejo d\omega \nonumber \\
& = & \frac{T}{2\pi}\int_0^{\frac{2\pi}{T}} \sum_{\ell=1}^N
\left|A_{\ell}\ejt \right|^2 d\omega,
\end{eqnarray}
where we used (\ref{eq:orth}).

\section{Proof of Proposition~\ref{prop:formm}}
\label{app:formm}

 To prove the proposition, we first
note that since $\phi_{\ell}(t)$ is in $\A$ for each $\ell$, we
can express it as
\begin{equation}
\label{eq:phip} \phi_{\ell}(t)=\sum_{r=1}^N \sum_{n \in \ZZ}
a_r^{\ell}[n]\psi_r(t-nT)
\end{equation}
for some coefficients $a_r^{\ell}[n]$ with Fourier transform
$A_r^{\ell}\ejo$. We have shown in the proof of
Theorem~\ref{thm:mu} that the orthonormality condition
(\ref{eq:orth}) of $\psi_\ell(t)$ implies that
\begin{equation}
\label{eq:ar}
A_r^{\ell}\ejo=\overline{R}_{\phi_{\ell}\psi_{r}}\ejo.
\end{equation}

Now, since $\{\phi_{\ell}(t-nT)\}$ is an orthonormal basis for
$\A$, $R_{\phi_\ell\phi_r}\ejo=\delta_{\ell,r}$. From
(\ref{eq:phip}),
\begin{eqnarray}
\label{eq:or} R_{\phi_\ell\phi_r}\ejo & = &
\sum_{m=1}^N\sum_{s=1}^N \overline{A}_m^{\ell}\ejo A_s^{r}\ejo
R_{\psi_m\psi_s} \nonumber \\
& = & \sum_{m=1}^N \overline{A}_m^{\ell}\ejo A_m^{r}\ejo
\nonumber \\
& = & [\bbm_{\phi\psi}\ejo]_\ell[\bbm_{\phi\psi}\ejo]_r^H,
\end{eqnarray}
where $[\bbc]_r$ denotes the $r$th row of $\bbc$. The second
equality in (\ref{eq:or}) follows from the orthonormality of
$\{\psi_{\ell}(t-nT)\}$, and the last equality is a result of
(\ref{eq:ar}). Since $R_{\phi_\ell\phi_r}\ejo=\delta_{\ell,r}$, it
follows from (\ref{eq:or}) that the matrix $\bbm_{\phi\psi}\ejo$
is unitary for all $\omega$.

Since $\bbm_{\phi\psi}\ejo$ is unitary, the coherence
$\mu(\bbi,\bbm_{\phi\psi}\ejo)$ is well defined. Now for any
unitary $\bbu$, $\mu(\bbi,\bbu)=\max_{i,j} |U_{ij}|$. In addition,
$\mu(\Phi,\Psi)=\max_{i,j} \sup_\omega
|[\bbm_{\phi\psi}\ejo]_{ij}|$, so that
$\mu(\bbi,\bbm_{\phi\psi}\ejo) \leq \mu(\Phi,\Psi)$, completing
the proof.


\end{document}